\documentclass[journal,onecolumn]{IEEEtran}
\usepackage{algorithm}
\usepackage{algpseudocode}
\usepackage{array}
\usepackage[caption=false,font=normalsize,labelfont=sf,textfont=sf]{subfig}
\usepackage{textcomp}
\usepackage{stfloats}
\usepackage{bm}
\usepackage{url}
\usepackage{amsmath, amsthm, amssymb, amsfonts, mathtools}
\usepackage{enumitem}
\usepackage{graphicx}
\usepackage{romannum}
\usepackage{multirow}
\usepackage{array} 
\usepackage{diagbox} 
\graphicspath{ {./images/} }
\usepackage{verbatim}
\usepackage{graphicx}
\usepackage{balance}
\usepackage{xcolor}
\newcommand{\Spdf}[3]{S\left( \alpha, #1, #2, #3 \right)}
\newcommand{\eqdef}{\,\hat{=}\,}

\makeatletter 
\@ifundefined{theorem}{%
\newtheorem{theorem}{Theorem}
\newtheorem*{theorem*}{Theorem}

\newtheorem{definition}{Definition}

\newtheorem{remark}{Remark}

\newtheorem{properties}{Properties}
}{}

\makeatother
\newcommand{\E}[1]{\mathsf{E}\left[{#1}\right]}   
\newcommand{\vect}[1]{\mathbf{#1}} 
\def\pz{f_{\tilde{Z}_{\alpha}} \!\!}

\def\vtheta{\boldsymbol\theta}
\def\sbs{\!}
\newcommand{\atilde}[1]{\tilde{#1}_{\alpha}\sbs}
\def\pzv{f_{\tilde{\vect{Z}}_{\alpha}} \!\!}

\def\sgn{\mathsf{sgn}}  
\def\RID{R^{\mathsf{(I)}\!}(D)}
\def\tRID{\tilde{R}^{\mathsf{(I)}\!}(D)}
\def\vx{\vect{x}}
\def\vX{\vect{X}}
\def\Xb{\textbf{X} }
\def\xb{\textbf{x} }
\def\Xhd{\hat{X}_{\!\scriptscriptstyle\Delta}}
\def\SaS{\textrm{S}$\alpha$\textrm{S} }
\newcommand{\Normal}[2]{\mathcal{N}\left({#1}, {#2}\right)}
\def\Reals{\mathbb R} 
\def\Integers{\mathbb Z} 
\newcommand{\set}[1]{\mathcal{#1}}   

\newcommand{\sa}[1]{s_{\alpha} \! \left( {#1} \right)}
\newcommand{\sat}[1]{\tilde{s}_{\alpha} \! \left( {#1} \right)}

\begin{document}
\title{A Framework for Lossy Compression of Heavy-Tailed Sources
}

\author{Karim Ezzeddine, Jihad Fahs$^*$,~\IEEEmembership{Member,~IEEE,} and Ibrahim Abou-Faycal,~\IEEEmembership{Member,~IEEE}
\thanks{This paper was presented in part at the 2025 IEEE International Symposium on
Information Theory. 

The authors are with the Department of Electrical and Computer
Engineering, American University of Beirut, Beirut, Lebanon
(e-mail: \{ke51,jihad.fahs,ibrahim.abou-faycal\}@aub.edu.lb).}}


\maketitle

\begin{abstract}
We study the rate-distortion problem for both scalar and vector memoryless heavy-tailed $\alpha$-stable sources ($0 < \alpha < 2$). Using a recently defined notion of ``strength" as a power measure, we derive the rate-distortion function for $\alpha$-stable sources subject to a constraint on the strength of the error and show it to be logarithmic in the strength-to-distortion ratio. We show how our framework paves the way for finding optimal quantizers for $\alpha$-stable sources and other general heavy-tailed ones. In addition, we study high-rate scalar quantizers and show that uniform ones are asymptotically optimal under the error-strength distortion measure. We compare uniform Gaussian and Cauchy quantizers and show that more representation points for the Cauchy source are required to guarantee the same quantization quality. Our findings generalize the well-known results of rate-distortion and quantization of Gaussian sources ($\alpha = 2$) under a quadratic distortion measure. 
\end{abstract}


\section{Introduction}
\IEEEPARstart{I}{n} 1959,  Shannon introduced source coding with a fidelity criterion~\cite{shannon1959coding}, also known as rate-distortion theory or lossy data compression. At its core, rate-distortion theory investigates the optimal trade-off between a given source's compression rate and the incurred distortion between that source and its reconstructed version. In his paper~\cite{shannon1959coding}, Shannon studied discrete alphabets and posited that everything could be extended to continuous alphabets with appropriate adjustments. Subsequently, a more elaborate treatment of the rate-distortion theory for continuous alphabets was made in~\cite{BERGER1968254} and more recently in~\cite{Polyanskiy_Wu_2024}. Naturally, rate-distortion theory proves beneficial across various domains when one has certain continuous measurements, as it aids in determining the attainability of a given rate of compression under some prescribed distortion and vice versa. Applications include, for example, identifying the fundamental limit on how much a machine learning model can be compressed~\cite{gao2019rate}. Another application can be found in~\cite{7728150} where the authors show, based on rate-distortion theory, how several stimulus-response curves that are frequently observed in biological signaling pathways arise naturally as the optimal decision strategy.

Beyond these applications, rate-distortion theory is tightly related to the concept of quantization in communications and signal processing, a process in which an analog source is represented by digital values. When quantizing any analog source, the goal is to find the minimum number of representation levels required to describe the source while maintaining the resulting distortion below some threshold level. It is therefore clear that rate-distortion theory provides the fundamental limits on achievable compression rates and hence the quantizer's size in bits given some distortion target and vice versa. Naturally, quantization finds applications in audio and image processing, robotics, sensor networks, where it aids in the representation and transmission of continuous data~\cite{quantimages,sensors}. In machine learning, quantization is used in large language models and deep neural networks to reduce memory usage~\cite{gong2014compressing,liu2025lowbitmodelquantizationdeep,xiao2024smoothquantaccurateefficientposttraining}. Additionally, new quantization techniques are constantly developed and implemented to improve the efficiency of systems with restrictions on hardware and storage capacity~\cite{hubara2016binarized,ieesp1,hardwarequantization}.

The Gaussian distribution has predominantly been used as a model for the continuous source in communications and information theory due to its unparalleled analytical tractability; the fact that it maximizes entropy among all distributions having finite second moments; and the Central Limit Theorem (CLT).  However, Gaussian sources are not universal source models: for example, Gaussian models are not appropriate whenever the source is better characterized by heavy-tailed statistics. In fact, power law distributions\footnote{A density function $f(x)$ is said to have a right (resp. left) tail power law if $\lim f(x)|x|^{\alpha + 1} = k$ for some $k > 0$ and $\alpha$, as $x \rightarrow +\infty$ (resp. $-\infty$).} are found in many real-world data sets~\cite{Tsih1995,crovella1998heavy,marcel2008,clauset2009power}. Under such scenarios, a family of probability distributions known as the $\alpha$-stable distributions\footnote{$\alpha$-stable distributions are power laws with $0 < \alpha < 2$.} represents a natural extension of the Gaussian law by virtue of a Generalized CLT (GCLT)~\cite{kolmo}. Moreover, it is found that $\alpha$-stable sources are good models for noise sources in a multitude of applications in~\cite{784467,nbrm2001,kuzer2004,nolancirc2013,he2014robust}. The authors in~\cite{bellido1993} show that the weights learned by backpropagation systematically deviate from Gaussianity. The results in~\cite{jipeng2024} uncover the heavy-tailed nature of neural networks weights having significant statistical dependence across layers and neurons. Moreover, the authors demonstrate that modeling weights with $\alpha$-stable laws better capture training dynamics and generalization properties. 
The $\alpha$-stable distributions appear to provide a more suitable model than the Gaussian one in various cases, making them reasonable candidates to investigate across different setups. 

The fact that $\alpha$-stable distributions possess infinite second moments~\cite{samorodnitsky1996stable} raises the following natural questions: What is an appropriate distortion measure when considering an $\alpha$-stable source in the rate-distortion problem? What are the optimal quantization points when quantizing an $\alpha$-stable random source or a heavy-tailed one in general? The absence of a finite second moment for \(\alpha\)-stable distributions makes it inappropriate to use the Mean Square Error (MSE) as a distortion measure. 
However, numerical  attempts are made as in~\cite{msealphastable} where the rate-distortion function for $\alpha$-stable sources is approximated using the Blahut-Arimoto algorithm \cite{blahut1972computation}. Alternative metrics to the MSE such as Mean Square Root Absolute Error (MSRAE) and Mean Absolute Error (MAE) are used as distortion measures for quantizing \(\alpha\)-stable sources in the range \( \alpha \in [1, 2] \) in~\cite{quantHeavy}. However, these distortion measures exhibit limitations: 
\begin{itemize}
\item[1-] They are valid for certain ranges of $\alpha$ (MAE is applicable when \(\alpha \geq 1\) and MSRAE is valid for \(\alpha > 0.5\))
\item[2-] These metrics suffer from a discontinuity property; the MAE metric for example changes from being finite  for $\alpha = 1 + \epsilon$, for some $\epsilon > 0$, to infinite  for $\alpha = 1 - \epsilon$. This is not natural as the source statistics remain ``approximately" the same.
\end{itemize}
Up to our knowledge, the literature lacks a unified theoretical analysis of the rate-distortion function of heavy-tailed sources. In this paper, we answer these questions by employing a {\em strength} quantity on the error as a distortion measure, a notion that was introduced in~\cite{fahs2017information}, and adopted later for the specific example of the Cauchy variable in~\cite{verdu} where the author investigated the rate-distortion function of a scalar Cauchy source. A key ingredient in the proof of~\cite{verdu} was in showing that the {\em strength} is convex for Cauchy mixtures. Unfortunately, the proof was not presented with sufficient details and is tightly related to the closed-form expression of the Cauchy Probability Density Function (PDF). Our main contributions are four-fold: 
\begin{itemize}
\item First, we generalize the Gaussian and Cauchy rate-distortion results to the family of symmetric $\alpha$-stable distributions by adopting a generic approach, one that does not rely on a ``potential" convexity of the {\em strength} measure. This possible convexity is challenging to analyze given the general non closed-form expressions of the symmetric $\alpha$-stable PDFs. 
\item Second, we derive the rate-distortion functions for $\alpha$-stable vector sources including the Cauchy case, where both the "sub-Gaussian" and independent $\alpha$-stable vectors are considered. While we do not consider in this work the rate-distortion problem for other heavy-tailed sources, the $\alpha$-stable distributions serve as a benchmark for compressing any heavy-tailed distribution, just as the Gaussian source serves as a benchmark for distributions with a finite second moment. This is due to the fact that both are central distributions —in the sense of the CLT and the GCLT, respectively.
\item Third, we showcase that our proposed framework enables a systematic approach to quantizing not only $\alpha$-stable sources but also any heavy-tailed source; a framework that overcomes the drawbacks of using fractional lower order moments like the MAE and the MSRAE. Using our framework,  we present an algorithm for designing a strength-optimal quantizer. We analyze the properties of such a quantizer and apply the proposed algorithm to find a strength-optimal quantizer for a Cauchy source. We use the found quantizer in the context of communicating under additive Cauchy noise and show an increase in transmission rates compared to those obtained when using a Gaussian quantizer with identical power.
\item Finally, we conduct an analytical study of the performance of uniform scalar quantizers in the high-rate regime under the strength measure. We also analyze non-uniform quantizers in the high-rate regime and show that the optimal solution coincides with the uniform quantizer. As such, we conclude that similarly to standard quantizers, there is little advantage in considering non-uniform quantizers over uniform ones at high rate. Through numerical analysis, 
we quantify the extra cost in terms of the number of representation points needed to achieve the same quantization quality for a Cauchy source compared to a Gaussian one. 
\end{itemize}

The rest of this paper is organized as follows: In Section~\ref{sc:Pre} we present stable distributions and the notion of strength. 
The rate-distortion results are presented in Section~\ref{sc:RD} where we state the scalar rate-distortion function of all stable sources under a strength constraint on the error in addition to extensions to the vector case. The proofs of the results are deferred to Appendix~\ref{app:proofs}. Section~\ref{sc:Quant} is dedicated to finding and applying strength-optimal quantizers and Section~\ref{sc:QuantHR} discusses the performance of uniform and non-uniform quantizers in the high-rate regime in addition to providing performance analysis through a numerical example. 
Finally, Section~\ref{sc:Concl} concludes the paper.  

{\bf \em Notations:} Lowercase regular font symbols $\{x, y, \cdots\}$ represent scalar deterministic quantities whereas uppercase ones $\{X, Y, \cdots\}$ are scalar Random Variables (RV)s. Lowercase bold font characters $\{{\bf x}, {\bf y}, \cdots\}$ are reserved for deterministic vector quantities and their uppercase counterparts $\{\vect{X}, \vect{Y}, \cdots\}$ are random vectors. We denote by $\Xb^n$ a sequence of $n$ vectors, i.e. $\Xb^n$  = $(\Xb_1, \cdots ,\Xb_n)$, where each $\Xb_k$ is a vector in $\Reals^d$. We write $\Xb \overset{d}{=} {\bf Y}$ whenever they have identical probability laws. The notation $\|.\|$ refers to the $L_2$ norm in $\mathbb{R}^d$.

\section{Preliminaries}
\label{sc:Pre}

We present hereafter stable random variables and sub-Gaussian symmetric alpha-stable vectors. We also formalize the notion of their ``strength".

\subsection{Stable Distributions}
\begin{definition}[Univariate Stable Distributions~{\cite[p.3 Definition 1.1.4]{samorodnitsky1996stable}}] A random variable $X$ is said to have a {\em stable\/} distribution if for any $k \geq 2$, there is a positive number $C_k$ and a real number $D_k$ such that
    \[
        X_1 + X_2 + \cdots + X_k \overset{d}{=} C_k X + D_k,
    \]
    where $\{X_1,X_2,\cdots,X_k\}$ are independent copies of $X$.
\end{definition}

\begin{definition}
    [Equivalent Definition~{\cite[p.5 Definition 1.1.5]{samorodnitsky1996stable}}] A random variable $X$ is said to be {\em stable\/} and denoted $X \sim \Spdf{\beta}{\gamma}{\delta}$ whenever its characteristic function is of the form
    \[
    \varphi_X(\omega) = \exp\left[i\delta\omega -  \gamma^\alpha\left(1 - i\beta \sgn(\omega)\Phi(\omega)\right)|\omega|^\alpha\right],
    \]
    where $0 < \alpha \leq 2$, $-1 \leq \beta \leq 1$, $\gamma > 0$, and $\delta \in \Reals$, and where
    $\sgn(\omega)$ is the sign of $\omega$ and $\Phi(\cdot)$ is given by:
    \[
    \Phi(\omega) = \begin{cases}
    \tan\left(\frac{\pi\alpha}{2}\right) & \text{if } \alpha \neq 1, \\
    - \frac{2}{\pi} \ln |\omega| & \text{if } \alpha = 1.
    \end{cases}
    \]
\end{definition}
Note that the case $\alpha = 2$ corresponds to a Gaussian RV. Whenever the Gaussian law is excluded, i.e., $\alpha < 2$ the RV is said to be {\em $\alpha$-stable\/}.

\begin{remark} When $X \sim S(\alpha,\beta,\gamma,\delta)$ with $\beta = \delta =0$, $\phi_X(\cdot)$ becomes $\varphi_X(\omega) = \exp\left[ -  \gamma^\alpha|\omega|^\alpha\right]$, $X$ is symmetric and we refer to it as {\em symmetric alpha-stable\/} (\SaS\!\!) and we write $X \sim S(\alpha,\gamma)$. 
\end{remark}

\begin{definition}[Sub-Gaussian\footnote{In some texts, the term sub-Gaussian refers to distribution functions whose
tails are faster than those of a Gaussian. In this work, sub-Gaussian \SaS is used in the sense of Definition~\ref{def:subgaus}.} Symmetric Alpha-Stable Vector~{\cite[p.78, Definition 2.5.2]{samorodnitsky1996stable}}] Let \(0 < \alpha < 2\) and let \(A \sim S \left(\frac{\alpha}{2}, 1, \cos(\frac{\pi \alpha}{4})^{\frac{2}{\alpha}}, 0\right) \) be a ``totally skewed" one-sided alpha-stable distribution. Define \(\vect{G} = (G_1, \cdots, G_d)^{\text{t}} \) to be a zero-mean Gaussian vector in \(\Reals^d\). Then the random vector
\[ \vect{N} = A^{\frac{1}{2}} \vect{G} = \left( A^{\frac{1}{2}} G_1, \cdots, A^{\frac{1}{2}} G_d \right)^{\text{t}}, \]
is called a {\em sub-Gaussian\/} Symmetric Alpha-Stable (\SaS\!\!) random vector in \(\Reals^d\) with underlying vector \(\vect{G}\). In particular, each component \(A^{\frac{1}{2}} G_i, \, 1 \leq i \leq d\) is a \SaS variable with characteristic exponent \(\alpha\).
\label{def:subgaus}
\end{definition}

Whenever the components 
of the underlying vector $\vect{G}$ are Independent and Identically Distributed (IID) zero-mean with variance $2 \gamma^2$ for some $\gamma > 0$, we will denote the \SaS random vector by $\textbf{S}(\alpha,\gamma)$. In \cite[p.79 Proposition 2.5.5]{samorodnitsky1996stable}, it is shown that the characteristic function of a vector $\vect{X} \sim \textbf{S}(\alpha,\gamma_X)$ is:
\begin{equation*}
    \varphi_{\vect{X}}(\vtheta) 
    = e^{- \gamma_X^\alpha \|\vtheta\|^\alpha}.
\end{equation*}

Several useful properties of stable distributions are listed in Property~\ref{prop} in Appendix~\ref{app:prop}.

\subsection{The Strength of a Random Variable / Vector}

Let $\atilde{\vect{Z}} \sim 
\textbf{S} \left(\alpha, \left( \frac{1}{\alpha} \right)^{\frac{1}{\alpha}} \right)$ be {\em a reference\/} $d$-dimensional \SaS vector, i.e.,
\begin{itemize}
    \item when $\alpha \neq 2$: a reference sub-Gaussian \SaS vector with underlying Gaussian vector having IID zero-mean components with variance $\sigma^2 = 2 \gamma^2 = 2 \left( \frac{1}{\alpha} \right)^{\frac{2}{\alpha}}$,
    \item when $\alpha = 2$: a reference Gaussian vector of IID components
    with mean zero and variance 1,
\end{itemize}
with $\pzv(\cdot)$ being its PDF and $h(\atilde{\vect{Z}})$ its differential entropy. We note that since \SaS densities are positive, unimodal, and have finite logarithmic moments (see property~\ref{prop}-2 in Appendix~\ref{app:prop}), then they have finite differential entropy~\cite{rioulEPI2011}.
\begin{figure}[!t]
    \begin{center}
    \includegraphics[width=0.489\linewidth]{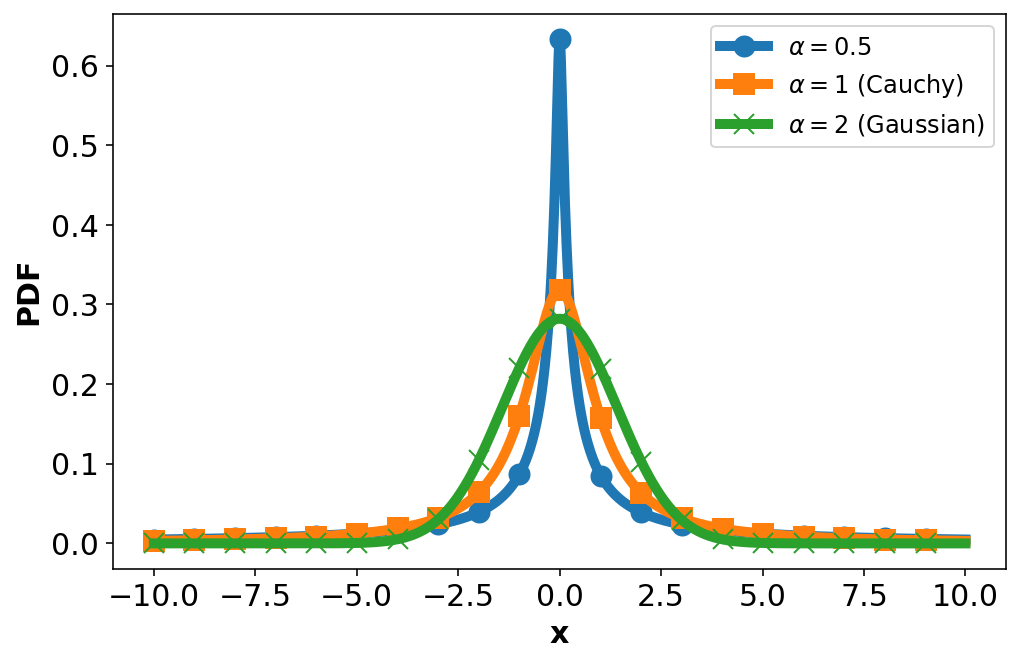}
\caption{A plot of $f_{X}(x)$ where $X \sim S(\alpha,1)$ for different values of $\alpha$.}
    \label{fig:stableDist}
    \end{center}
\end{figure}

 
\begin{definition}\cite[Definition 3]{fahs2017information}
\label{def:str}
The $\alpha$-power or strength of a $d$-dimensional random vector \( \vX \) is:
\begin{equation}
\sa{\vX} \eqdef \inf \left\{ s>0 : -\E{ \log \pzv \left( \frac{\vX}{s} \right)}  \leq h(\atilde{\vect{Z}}) \right\}.
\label{eq:powdef}
\end{equation}
\end{definition}

Note that when \( \vect{X} \) is deterministically equal to the zero vector, $ s_\alpha(\vX)=0 $. Indeed, for $\Xb = {\bf 0}$, any $s > 0$ would satisfy the constraint \eqref{eq:powdef} since $-\log \pzv \left({\bf  0} \right) \leq -\log \pzv \left( \vx \right)$ for all $\vx \in \Reals^d$.

Several basic properties of this strength measure are listed in Property~\ref{prop:stableprop} in Appendix~\ref{app:prop}. It is worth noting that Property~\ref{prop:stableprop}-1 states that, unless the variable is deterministically equal to zero, $s = s_\alpha(\vX)$ achieves the equality in~\eqref{eq:powdef}.

For the two special cases where $\pzv(\cdot)$ has a closed-form expression, we get:
\begin{itemize}
    \item When \( \alpha =2 \)~\cite{fahs2017information}, \( \tilde{\vect{Z}}_2 \) is a {\bf 0}-mean Gaussian vector with independent components with unit variance each and 
    \begin{equation}
    \label{eq:stren2}
    s_2(\vect{X}) = \sqrt{\frac{\E{\|\vect{X}\|^2}}{d}}.
    \end{equation}
    This is consistent with the view that the strength is a generalization of the standard deviation for variables lacking a finite second moment.
    \item When \( \alpha =1 \)~\cite{verdu}, \( \tilde{\vect{Z}}_1 \) is a ``standard" circular Cauchy random vector $\textbf{S} \left( 1, 1 \right)$ and \( s_1(\vect{X}) \) is the unique constant such that
    \begin{equation}
    \E{\ln\left(1+\frac{\|\vect{X}\|^2}{s_1(\vect{X})^2}\right)} = \ln(4) + \psi \left( \frac{d+1}{2} \right) + \gamma_{\text{e}},
    \label{eq:stren1}
    \end{equation}
    where $\gamma_{\text{e}}$ is the Euler–Mascheroni constant and $\psi$ is the digamma function. 
\end{itemize}
\begin{remark} Consider a random vector $\vX$ with strength $0<s_\alpha({\bf X})<A$ for some $A>0$. Then, $A\tilde{{\bf Z}}_\alpha = \underset{\vX}{\text{argmax}}\ h\left( \vX \right)$. 
\end{remark}
Proven in~\cite[Theorem 1]{fahs2017information}, this property shows that among all distributions having their strength upper bounded by some constant $A$, the \SaS distribution $A \tilde{{\bf Z}}_\alpha$ maximizes entropy. This is a natural extension to the well-known result stating that Gaussian distributions maximize entropy among all distributions having an upper bounded second moment.

\section{The Rate-distortion Function for Stable Sources}
\label{sc:RD}

In this section, we present the rate-distortion results for sources $\Xb \sim \textbf{S}(\alpha,\gamma_{\vX})$, for all $\alpha \in (0,2]$. Since an MSE distortion measure is sensible {\em only\/} for $\alpha = 2$, we measure the distortion between $\vX$ and its representation $\hat{\vX}$ through the strength of the error: $\sa{\vX-\hat{\vX}}$.
This setup deviates from the classical one in the fact that the distortion --measured in strength-- directly acts on random variables and not through sample values. This results in a couple of challenges in proving the rate distortion theorem:
\begin{itemize}
\item Achievability: The direct part of the theorem will require defining a new ``equivalent" distortion constraint that is sample-value based.   
\item The converse part will require working with the expression of the rate-distortion function explicitly, instead of relying on the convexity of the constraint in the distributions; Said differently, the convexity of the rate-distortion function cannot be established a priori. When it comes to finding the expression of the rate-distortion function, we will rely on the fact that stable distributions are entropy maximizers under a strength constraint.
\end{itemize}

\subsection{The Rate-Distortion Problem}

Consider a memoryless source $\Xb^{n} \overset{\textup{i.i.d.}}{\sim} P_\Xb$. We consider the rate-distortion problem where we need to find the minimal rate $R$ of compression for the source under the constraint 
\begin{equation}
\sa{ \Xb-\hat{\Xb} } \leq D,
\label{eq:DistTB}
\end{equation}
where $\hat{\Xb}$ is the reconstruction of $\Xb$ and $s_\alpha(\cdot)$ is used to measure the distortion between the source and its reconstruction.
 
To make the statement formal, we introduce the following quantities in parallel to what is presented in~\cite{Polyanskiy_Wu_2024}: 
\begin{enumerate}[label=\textbullet]
    \item We measure the distortion between a length-$n$ sequence of the source ${\bf X}^n$ and its representation sequence ${\bf \hat{X}}^n$ as the time-average of the map $s_\alpha(\cdot)$ given by definition \ref{def:str}: 
    \[
    \frac{1}{n} \sum_{k=1}^{n} \sa{\Xb_k - \hat{\Xb}_k}.
    \]
    \item An $(n,M,D)$-code for $\Xb^n$ consists of 
    \begin{itemize}
        \item[-] a compressor $f_n:\bm{\chi}^{n} \to \set{M} \eqdef \{1,2,\cdots,M\}$ 
        \item[-] and a decoder $g_n: \set{M} \to \hat{\bm{\chi}}^n$,
    \end{itemize}
    such that $\frac{1}{n} \sum_{k=1}^{n} \sa { \Xb_k - g(f({\Xb^n}))_k } \leq D$. Note that $\bm{\chi}$ and $\hat{\bm{\chi}}$ are the source and its reconstruction alphabets respectively.
    \item The rate-distortion function is defined as 
    \begin{align*}
        R(D) & \eqdef \, \underset{n\to \infty}{\lim \, \sup}\ \frac{1}{n} \log M^{*}(n,D) \\
        \text{where } \quad M^{*}(n,D) & \eqdef \min \, \big\{ M: \exists \,\,(n,M,D)\text{-code} \big\}.
    \end{align*}
    \item The information rate-distortion function is defined as: 
    \begin{align*}
        \RID & \eqdef \underset{n\to \infty}{\lim \, \sup}\ \frac{1}{n} \, \phi_{\Xb^{n}}(D) \\
        \text{where } \quad \phi_{\Xb^n}(D) & \eqdef  \underset{\substack{P_{\hat{\Xb}^n|\Xb^n}: \\ \frac{1}{n} \sum_{1}^{n} \sa{ \Xb_k - \hat{\Xb}_k} \leq D}}{\inf} I\left(\Xb^n;\hat{\Xb}^n\right).
    \end{align*}
\end{enumerate} 
In contrast, in the standard formulation of the rate-distortion problem, given by a map $d(\cdot,\cdot)$:
\begin{equation*}
    d: \bm{\chi} \times \hat{\bm{\chi}} \to \Reals,
\end{equation*}
the distortion between a length-$n$ source sequence of characters ${\bf x}^n$ and its representation sequence ${\bf \hat{x}}^n$ is defined as \\ $d(\xb^n, \hat{\xb}^n) \eqdef \frac{1}{n} \sum_{k=1}^{n} d(\xb_k, \hat{\xb}_k)$. This property is referred to as the separability of the distortion measure in \cite{Polyanskiy_Wu_2024} where the distortion between two sequences of characters is equal to the time average of the per-character distortions. Naturally, the goal is to find the minimal compression rate for the source under the constraint:
\begin{equation*}
    \E{d \left( \Xb^n, \hat{\Xb}^n \right)} \leq D,
\end{equation*}
meaning that for an $(n,M,D)$-code, the compressor and the decoder must be such that $\E{d(\Xb^{n}, g(f(\Xb^{n})))} \leq D$. The remaining definitions, specifically for the rate-distortion function and the information rate-distortion function, are identical, with the exception that $d(\cdot,\cdot)$ replaces $s_\alpha(\cdot)$.
We now state the rate-distortion theorem due to Shannon \cite{shannon1959coding}, which was particularly well-defined and proven for continuous alphabets in \cite{Polyanskiy_Wu_2024}.

\begin{theorem}[Rate-Distortion Theorem]
\label{th:RDT}
Let $X^{n} \overset{\textup{i.i.d.}}{\sim} P_X$ be a stationary memoryless source, $d: \chi \times \hat{\chi} \to \Reals$ be a separable distortion metric and $D >0$ be the target distortion. If:
\begin{enumerate}
    \item $d(\cdot,\cdot)$ is non-negative 
    \item $D > D_0 \eqdef \inf \, \big\{ D : \phi_{X}(D) < \infty \big\}$
    \item $ D_{\max} \eqdef 
    \inf 
    \, \big\{ \E{d(X,\hat{x})}, \hat{x} \in \hat{\chi} \big\}$ is finite
    \item $\phi_X(D)$ is convex in $D$,
\end{enumerate}
then,
\begin{align}
    R(D) = \RID=\phi_{X}(D).  \notag
\end{align}
\end{theorem}
\begin{proof}
See \cite[Theorem 25.1]{Polyanskiy_Wu_2024}
\end{proof}

We highlight the fact that  the convexity of $\phi_X(D)$ is not explicitly stated as a requirement in~\cite[Theorem 25.1]{Polyanskiy_Wu_2024}; For under a standard ``textbook" formulation of the distortion constraint~\eqref{eq:DistTB}, the linearity of $\E{d(X,\hat{X})} =\sum_{(x,\hat{x})}P_X(x)P_{\hat{X}|X}(\hat{x}|x)d(x,\hat{x})$ in the distributions and the convexity of $I(P_X;P_{\hat{X}|X})$ in $P_{\hat{X}|X}$ imply that $\phi_X(D)$ is always convex in $D$~\cite[p. 316-317]{Cover2006}. 

Note that convexity of $\phi_X(D)$ is crucially used  in standard proofs of the converse $R(D) \geq \phi_X(D)$: Indeed, one can first show that $R(D) \geq \RID$ \cite[Corollary 24.6]{Polyanskiy_Wu_2024}, while the convexity of $\phi_X(D)$ yields $\RID \eqdef \underset{n\to \infty}{\lim \ \sup}\ \frac{1}{n} \phi_{X^{n}}(D) = \phi_X(D)$.


\subsection{The Rate-Distortion Function}

We study the rate-distortion problem for stable sources $\Xb$ 
for all $\alpha \in (0,2]$. Since an MSE distortion measure is sensible {\em only\/} for $\alpha = 2$, we measure the distortion between a length-$n$ sequence of the source ${\bf X}^n$ and its representation sequence ${\bf \hat{X}}^n$ by the strength of the error defined as $\frac{1}{n} \sum_{k=1}^{n} \sa{ \Xb_k - \hat{\Xb}_k}$, as indicated previously. 

Note first that adopting in the scalar case $\sa{X-\hat{X}}$, for $\alpha \in (0,2]$ as a distortion measure is a generalization of the following previously-studied two cases: 
\begin{enumerate}
    \item $\alpha = 2$: in which case $X \sim S(2,\gamma) \equiv \Normal{0}{\sigma^2 = 2 \gamma^2}$ and $s_2(X - \hat{X}) \equiv \sqrt{E{\left[|X-\hat{X}|^2\right]}} \leq D$. It is well known that in this scenario,
    \begin{equation}
    R(D)=\frac{1}{2} \log \left[ \frac{\sigma^2}{D^2} \right] = \frac{1}{2} \log \left[ \frac{2 \gamma^2}{D^2} \right] \label{eq:stdG}.
    \end{equation}
    \item $\alpha = 1$: is the setup studied in~\cite{verdu} where the rate-distortion function of a centered Cauchy source is derived.
\end{enumerate}
In what follows, we state our main generalization of the rate-distortion function to all scalar symmetric $\alpha$-stable sources, then we extend the result to vector cases. The proofs are deferred to Appendix~\ref{app:proofs}. 

\begin{theorem}[Extension to symmetric stable sources]
\label{th:extstable}
Consider a stationary memoryless \SaS source $\{X_n\}_n \overset{\textup{i.i.d.}}{\sim} P_X = S(\alpha,\gamma_X)$ with strength $s_\alpha(X) =  (\alpha)^{\frac{1}{\alpha}} \gamma_X$. The rate-distortion function of $X$ is given by:
\begin{equation}
\label{eq:star}
R(D) = \mathcal{R}_X(D) = \max \left\{ \log \left[ \frac{s_\alpha(X)}{D} \right], 0 \right\}.
\end{equation}
\end{theorem}

We note that whenever $\alpha = 2$, $X \sim S(2,\gamma_X)$ 
and $s_2(X-\hat{X})=\sqrt{\E{|X-\hat{X}|^2}}$, the results of Theorem~\ref{th:extstable} imply that $R(D) = \log \left[ \frac{s_2(X)}{D} \right] = \log\left(\frac{\sqrt{2} \gamma_X}{D}\right)= \frac{1}{2}\log\left( \frac{\sigma_X^2}{D^2}\right)$ which coincides with~(\ref{eq:stdG}). 
In addition, the rate-distortion function evaluated in~\cite{verdu} for the Cauchy source ($\alpha=1$) coincides with \eqref{eq:star} as well. Figure~\ref{fig:RD} depicts $R(D)$ of various stable sources for different values of the parameter $\alpha$. 

\begin{figure}[!t]
    \begin{center}
    \includegraphics[width=0.489\linewidth]{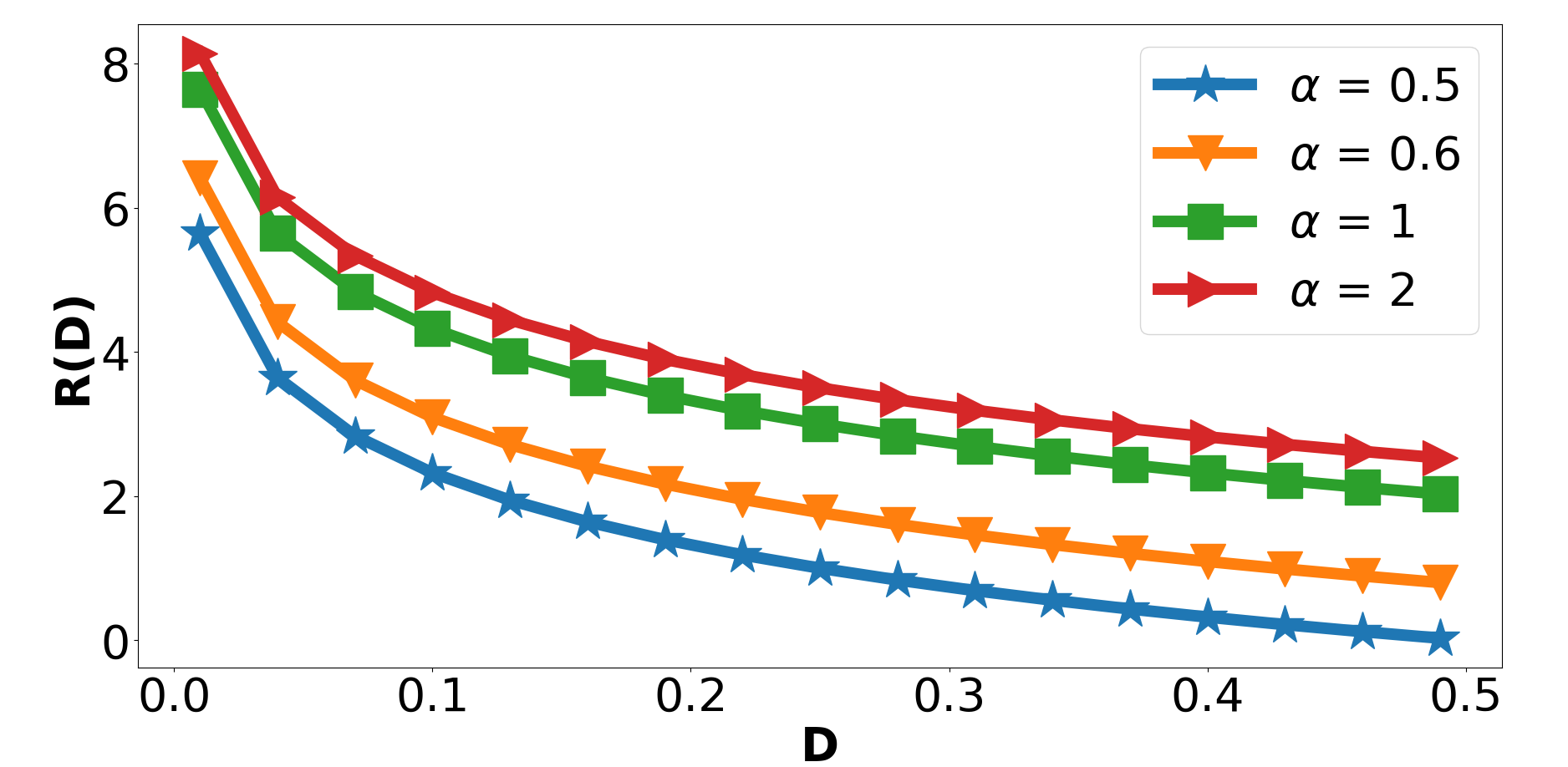}
\caption{The rate-distortion function $R(D)$ whenever $X \sim S(\alpha,2)$ for different values of $\alpha$.}
    \label{fig:RD}
    \end{center}
\end{figure}

\subsection{Extension to Vectors}
We extend the results of Theorem~\ref{th:extstable} to two types of $d$-dimensional stable vector-valued sources.

\begin{theorem}[Extension to {vectors} \Romannum{1}]
\label{th:extToVectors1}
Consider a stationary memoryless sub-Gaussian \SaS source $\{\Xb_n\}_n \overset{\textup{i.i.d.}}{\sim} P_{\Xb}  = \textbf{S}(\alpha,\gamma_{\Xb})$ with underlying $d$-dimensional Gaussian vector \textbf{G} having IID components. Then $s_\alpha(\Xb) =  (\alpha)^{\frac{1}{\alpha}} \gamma_\Xb$, and the rate-distortion function of $\Xb$ is given by:
\begin{equation}
\label{eq:starstar}
R(D)= \max \left\{ \log \left[ \frac{s_\alpha(\Xb)}{D} \right],0 \right\}.
\end{equation}
\end{theorem}

\begin{theorem}[Extension to {vectors} \Romannum{2}]
\label{th:extToVectors2}
Consider a stationary memoryless source $\{\Xb_n\}_n \overset{\textup{i.i.d.}}{\sim} P_{\Xb}$, where \Xb has $d$-independent \SaS components $(X_1, \cdots, X_d)$ with respective strengths $s_\alpha(X_i)$, $1 \leq i \leq d$. Whenever the distortion constraint is\\ $\sum_{j=1}^{d} \left[ \frac{1}{n} \sum_{k=1}^{n} \sa{  {X_k}_j-\hat{X}_{kj} } \right] \leq D$, The rate-distortion function of $\Xb$ is given by:
\begin{equation}
\label{eq:starstarstar}
R(D)=  \sum^{d}_{i=1} \log \left[ \frac{s_\alpha(X_i)}{D_i} \right],
\end{equation}
where
\[
D_i= 
    \begin{cases}
    \lambda, & \text{if} \quad \lambda < s_\alpha(X_i) \\
    s_\alpha(X_i), & \text{if} \quad \lambda \geq s_\alpha(X_i),\\  
    \end{cases}
\]
and $\lambda$ is chosen such that $\sum^{d}_{i=1} D_i = D$.
\end{theorem}
\section{The Quantization Problem}
\label{sc:Quant}

We consider scalar sources $X$ with arbitrary  unimodal and symmetric PDFs hereafter. Without loss of generality, we assume symmetry around 0. We define a new family of scalar quantizers using $s_\alpha(X)$ as a quality measure. This lays down a framework for analyzing the performance of Analog-to-Digital Converters (ADC) for sources that are heavy-tailed. Furthermore, we devise an algorithm that finds the representation points and the quantization regions for the optimal $M$-points quantizer and compute the resulting distortion.     

We provide in what follows a systematic optimal quantizer-design procedure for our setup and we present some results using numerical evaluations, followed by some remarks.

Quantization consists in representing a --possibly continuous-- source $X$ by a discrete random variable $\hat{X}$ taking values among $M$ quantization levels $\{\hat{x}_0,\cdots,\hat{x}_{M-1}\}$ while maintaining some prescribed fidelity criteria. Practically, a quantizer maps $x \in \Reals$ to $\hat{x}_j \in \Reals$ whenever $x \in \set{R}_j = (r_j,r_{j+1}]$. The $\{\set{R}_j\}$'s are a partition of $\Reals$ and define the {\em quantization regions\/}, and the $\{\hat{x}_j\}$'s are the {\em representation points\/}. Naturally, we are interested in finding an optimal $M$-points quantizer (equivalently $M$ quantization regions and $M$ representation points) of a heavy-tailed source, i.e. one that minimizes $\sa{ X - \hat{X}}$: Given $M > 0$, find  
 \begin{equation*}
        \underset{\hat{x}_0, \cdots, \hat{x}_{M-1},\set{R}_0,\cdots, \set{R}_{M-1} }{\arg \min} \, \sa{ X - \hat{X}},
\end{equation*}
where $\sa{ X - \hat{X}}$ is the strength of the error and hence solution of
\begin{multline}
    -\int_{\set{R}_0} f_X(x)\log \left[ \pz\left(\frac{x - \hat{x}_0}{\sa{ X - \hat{X}}}\right)\right] dx
    - \int_{\set{R}_1} f_X(x) \log \left[\pz\left(\frac{x - \hat{x}_1}{\sa{ X - \hat{X}}}\right)\right] dx -\ldots \, \\
     - \int_{\set{R}_{M-1}} \hspace{-7pt} f_X(x) \log \left[ \pz\left(\frac{x - \hat{x}_{M-1}}{\sa{X - \hat{X}}}\right)\right] dx  = h(\tilde{Z}_\alpha). \label{eq:quancons}
\end{multline}

\smallskip
\noindent
Before proceeding, we note that since the PDF of \( X \) is unimodal and symmetric around $0$, then whenever \( M \) is even, half of the representation points will be positive, and the other half will be their negative counterparts. If \( M \) is odd, the same applies, but with an additional representation point at $0$. Therefore, we only search for the positive representation points, as the remaining ones can be easily deduced.

\subsection{Optimality Conditions for Strength-Quantizers}

One can determine necessary conditions of optimality by "fixing" the representation points and studying the quantization regions and vice-versa.

With the classical MSE distortion measure, these necessary conditions are found to be:
\begin{itemize}
\item[(i)] The boundaries of the {\em quantization regions\/} are defined by the midpoints between one representation point and the next. 
\item[(ii)] {\em Representation points\/} are conditional means.
\end{itemize}

When it comes to the proposed strength distortion measure,
we show that property (i) is also a necessary condition for the optimal quantizer. In fact, let $\{\hat{x}_0,\cdots,\hat{x}_{M-1}\}$ be some randomly chosen $M$ quantization points on the real line and choose for all $ 0 \leq j \leq (M - 2)$, $r_j = \frac{\hat{x}_j + \hat{x}_{j+1}}{2}$. Under these fixed parameters, the value of $\sa{X - \hat{X}}$ is determined using Equation~\eqref{eq:quancons}. We show that if any of the $\{r_j\}$'s takes any other value than the midpoint, the resulting power $\sat{X - \hat{X}}$ will necessarily increase. Indeed, choose an arbitrary $0 \leq j \leq (M -2)$ and replace $r_{j}$ by $\tilde{r}_{j} \neq r_j$, such that $\tilde{r}_{j} \in [\hat{x}_{j},\hat{x}_{j+1}]$. We assume without loss of generality that $\tilde{r}_{j} > r_j \left( = \frac{\hat{x}_{j} + \hat{x}_{j+1}}{2} \right),$ and update the decision regions accordingly. Two terms on the left side of~\eqref{eq:quancons} will change with the updated quantization regions to become
\begin{align*}
    & \begin{multlined}[t]
    - \int_{r_{j-1}}^{\tilde{r}_{j}} f_X(x)\log \left[ \pz\left(\frac{x - \hat{x}_j}{\sa{X - \hat{X}}}\right)\right] \, dx 
     - \int_{\tilde{r}_{j}}^{r_{j+1}} f_X(x)\log \left[ \pz\left(\frac{x - \hat{x}_{j+1}}{\sa{X - \hat{X}}}\right) \right]\, dx
    \end{multlined} \\
    = & \begin{multlined}[t]
    - \int_{r_{j-1}}^{{r}_{j}} f_X(x)\log \left[ \pz\left(\frac{x - \hat{x}_{j}}{\sa{X - \hat{X}}}\right)\right] dx
    - \int_{r_{j}}^{\tilde{r}_{j}} f_X(x)\log \left[ \pz\left(\frac{x - \hat{x}_{j}}{\sa{X - \hat{X}}}\right)\right] dx \\
    - \int_{\tilde{r}_{j}}^{r_{j+1}} f_X(x)\log \left[ \pz\left(\frac{x - \hat{x}_{j+1}}{\sa{X - \hat{X}}}\right)\right] dx,
    \end{multlined}
\end{align*}
versus the original quantization rule
\begin{multline*}
    -\int_{r_{j-1}}^{r_{j}} f_X(x)\log \left[ \pz\left(\frac{x - \hat{x}_j}{\sa{X - \hat{X}}}\right)\right] dx 
    - \int_{r_{j}}^{\tilde{r}_{j}} f_X(x) \log \left[\pz\left(\frac{x - \hat{x}_{j+1}}{\sa{X - \hat{X}}}\right)\right] dx \\
    - \int_{\tilde{r}_{j}}^{r_{j+1}} f_X(x)\log \left[ \pz\left(\frac{x - \hat{x}_{j+1}}{\sa{X - \hat{X}}}\right)\right] dx.
\end{multline*}

 We note that only the second term has changed. Since $\forall x \in [r_{j},\tilde{r}_{j}]$, $|x - \hat{x}_{j}| >  |x - \hat{x}_{j+1}|$, then using the fact that $\pz(\cdot)$ is decreasing in $|x|$ (Property~\ref{prop}-3 in the appendix) we get
 \begin{align*}
    -\log\left[\pz\left(\frac{x - \hat{x}_{j+1}}{\sa{X - \hat{X}}}\right)\right] & <  -\log\left[\pz\left(\frac{x - \hat{x}_{j}}{\sa{X - \hat{X}}}\right)\right]
\end{align*}
\begin{equation*}
    \implies  -\int_{r_{j}}^{\tilde{r}_{j}} f_X(x) \log\left[\pz\left(\frac{x - \hat{x}_{j+1}}{\sa{X - \hat{X}}}\right)\right] \, dx < -\int_{r_{j}}^{\tilde{r}_{j}} f_X(x)\log\left[ \pz\left(\frac{x - \hat{x}_{j}}{\sa{X - \hat{X}}}\right)\right] \, dx.
 \end{equation*}
 
Therefore, the Left-Hand Side (LHS) of Equation~\eqref{eq:quancons} becomes bigger than $h(\tilde{Z}_{\alpha})$ which necessarily implies that $\sat{X - \hat{X}} > \sa{X - \hat{X}}$ in order to maintain equality in~\eqref{eq:quancons} by virtue of Property~\ref{prop:stableprop}-5 in the appendix. This proves that property (i) is a necessary condition of optimality.

When it comes to property (ii), there is no equivalent condition for strength-optimal quantizers. This is due to the fact that the strength is not ``decomposable" over the various quantization regions.

\subsection{Numerical Quantizer Design}

To design an optimal quantizer for $X$, we use a Lloyd-Max~\cite{LLoydMax} like algorithm where we iterate between optimizing given the representation points and optimizing given the quantization regions:
\begin{itemize}[label = $\longrightarrow$]
    \item Step 1: We randomly initialize $\left\lfloor \frac{M}{2} \right\rfloor$ representation points on the positive axis (in addition to a point at $0$ if $M$ is odd), and their symmetric negative counterparts.
    \item Step 2: Given the representation points, update the quantization regions according to the midpoint rule.
    \item Step 3: Given the quantization regions, we update the representation points $\{\hat{x}_0,\cdots,\hat{x}_{M-1}\}$ by minimizing $\sa{X - \hat{X}}$ such that Equation~\eqref{eq:quancons} is satisfied.
    \item Step 4: Repeat Steps 2 and 3 until the improvement in $\sa{X - \hat{X}}$ between one iteration and the next is small enough. 
\end{itemize}

\begin{algorithm}[!ht]
   \caption{Quantizer - Design}
   \label{alg:1}
    \begin{algorithmic}[1]
   \State {\bfseries Input:} n (even), $\alpha$, \text{tol}
   \Function{Power}{$S, \hat{\set{X}}$}
   \State \Return $S$
   \EndFunction
   \Function{Constraint}{$S, \hat{\mathcal{X}}, \mathcal{R}$}
   \State Compute $\displaystyle I_i = - \int_{r_i}^{r_{i+1}} f_X(x) \log \left[ f_{\tilde{z}_\alpha}\left(\frac{x - \hat{x}_{i}}{S}\right)\right] dx$ 
   \State Compute $\displaystyle I = \sum_{i=1}^{\left\lfloor \frac{M}{2} \right\rfloor-2} I_i$
   \State \Return $2I \leq h(\tilde{Z}_\alpha)$
   \EndFunction 
    \Function{Minimize}{$S, \hat{\mathcal{X}}, \mathcal{R}$}
    \State $\text{Bounds} \gets \left( (0, e_1), (e_1, e_2), \dots, (e_{\left\lfloor \frac{M}{2} \right\rfloor - 1}, +\infty) \right)$
    \State $\text{options} \gets \text{setTolerance}$
    \State $f \gets \Call{Power}{S, \hat{\mathcal{X}}}$
  \State $c \gets \Call{Constraint}{S, \hat{\mathcal{X}}, \mathcal{R}}$
 \State \Return \Call{ARGMIN}{f, bounds = \texttt{Bounds}, constraints = $c$, method = \texttt{trust-constr}}
   \EndFunction
    \State Randomly select $\hat{\set{X}} = \left\{\hat{x}_0, \hat{x}_1, \cdots, \hat{x}_{\left\lfloor \frac{M}{2} \right\rfloor-1} \right\}$  
    \Repeat
    \State Set $r_j = \frac{\hat{x}_j + \hat{x}_{j+1}}{2}$. $\set{R} = \left\{r_0, r_1, \cdots, r_{\left\lfloor \frac{M}{2} \right\rfloor-2} \right\}$ 
    \State Set $\{\hat{x}_j\}$ =  \Call{Minimize}{$S, \hat{\set{X}}, \set{R}$}
    \State Compute \Call{Power}{}
    \Until{$\big| \text{Power}_{\text{old}} - \text{Power}_{\text{new}} \big| < \text{tol}$}
\end{algorithmic}
\end{algorithm}

The corresponding pseudo-code is shown in Algorithm~\ref{alg:1}. In our implementation, we used a built-in constrained minimization function \textproc{Minimize} (denoted \textproc{ARGMIN} in the algorithm) from Python with the \texttt{trust-constr} method. This function is used to return the location of the quantization points that minimize the strength of the error. We emphasize that convergence is guaranteed since $s_\alpha(X - \hat{X}) > 0$ decreases in each repetition of Steps 2 and 3. 

We use Algorithm~\ref{alg:1} to design a quantizer for  a centered Cauchy random variable $X \sim S(1,\gamma_X)$ with $M = 2, 3$, and $4$ quantization points and for $\gamma_X = 1, 2, \cdots, 10$. As expected, the strength of the error decreases when the number of quantization points increases, and increases with $\gamma_X$ for fixed $M$ as clearly inferred form Figure~\ref{fig:RD1}.

\begin{figure}[!t]
\centering
\subfloat[]{\includegraphics[width=0.489\textwidth,keepaspectratio]{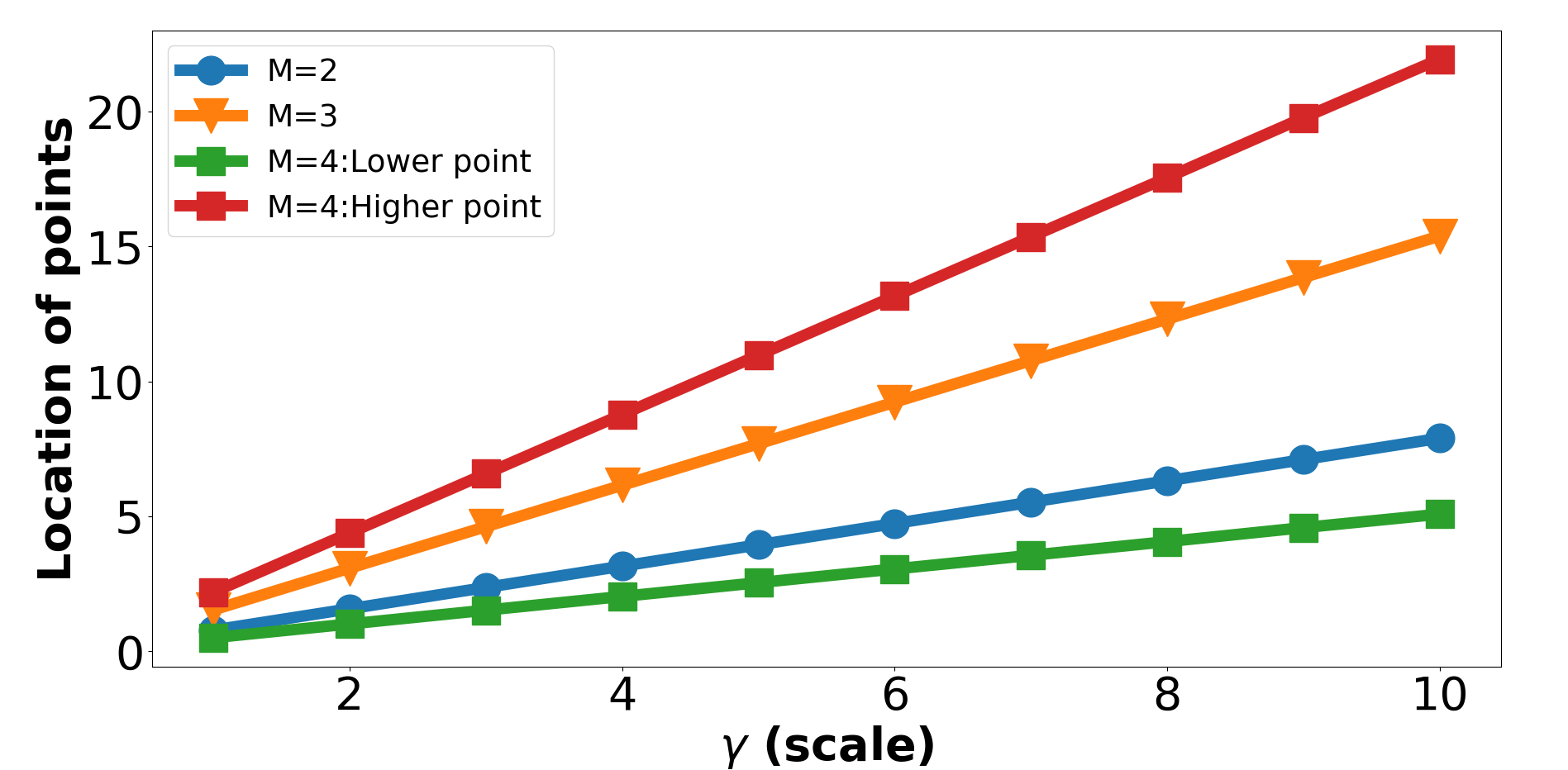}}%
\hfill
\subfloat[]{\includegraphics[width=0.489\textwidth, keepaspectratio]{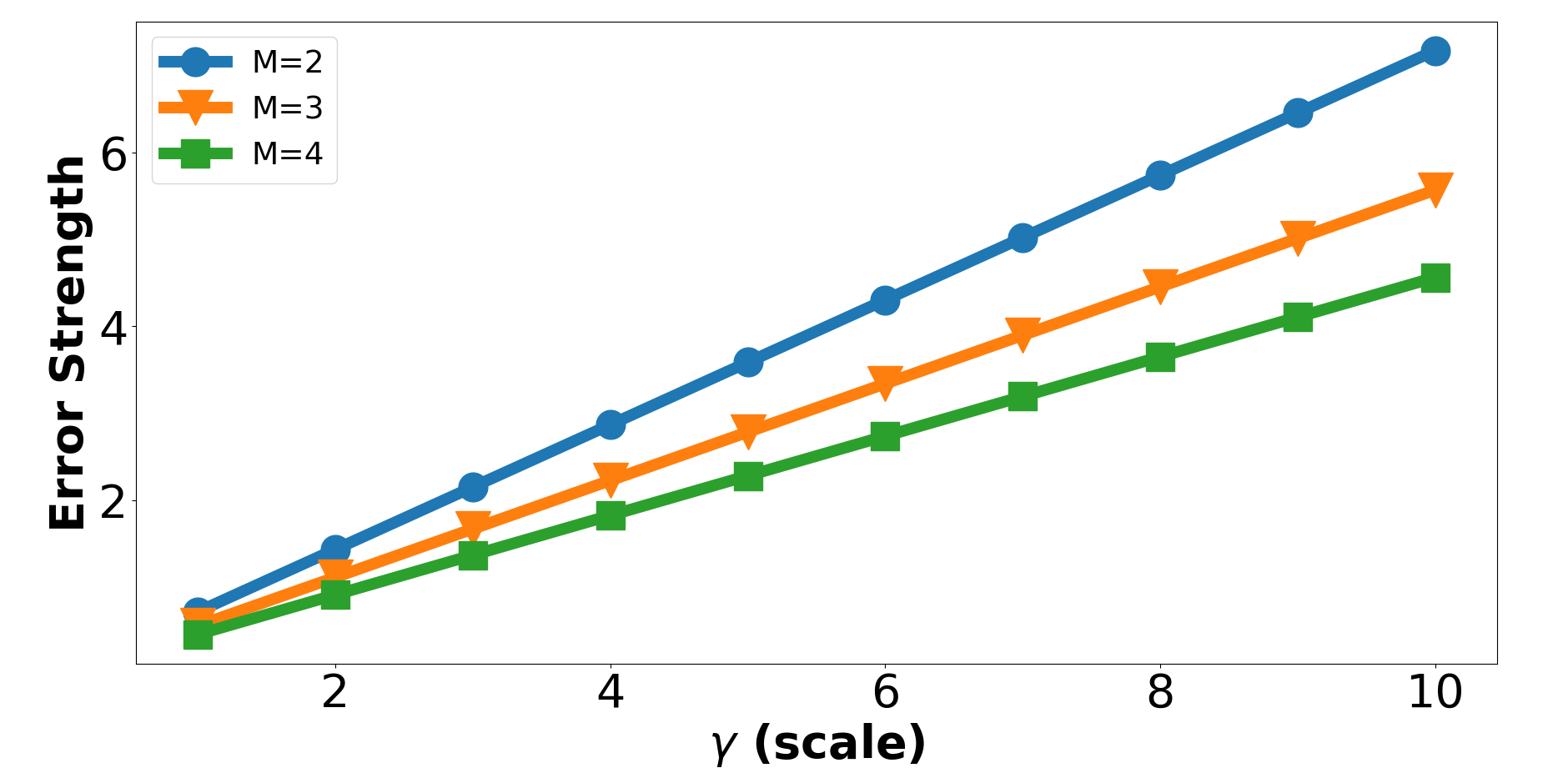}}%
\caption{Optimal $M = 2, 3,$ and $4$-points quantizers whenever $X \sim S(1,\gamma_X)$ for different values of $\gamma_X$. (a) Locations of the positive representation points; negative points are symmetric. (b) Resulting minimum strength.}
\label{fig:RD1}
\end{figure}
We also consider a source $X$ that follows a Student's $t$-distribution $X \sim \mathcal{T}_\nu(\gamma)$ with degrees of freedom $\nu$ and scale parameter $\gamma$ with PDF:
$$f_X(x) = \frac{1}{\gamma} f_{\mathcal{T}_\nu}\left(\frac{x}{\gamma}\right),$$ 
where:
\begin{equation*}
f_{T_\nu}\left(x\right) = \frac{\Gamma\left(\frac{\nu + 1}{2}\right)}{\sqrt{\pi} \Gamma\left(\frac{\nu}{2}\right)} \left(1 + \frac{x^2}{\nu}\right)^{- \frac{\nu + 1}{2}}. 
\end{equation*}
We determine an optimal 3-points quantizer for this source for values of $\nu \in \{0.7,1, 1.3, 1.9\}$. For each distribution, we plot in Figure~\ref{fig:RD2} --function of the scale $\gamma$, the location of the positive point (the other two being 0 and the negative counterpart of the positive point) along with the resulting error strength obtained by our quantizer.
\begin{figure}[!t]
\centering
\subfloat[]{\includegraphics[width=0.489\textwidth,keepaspectratio]{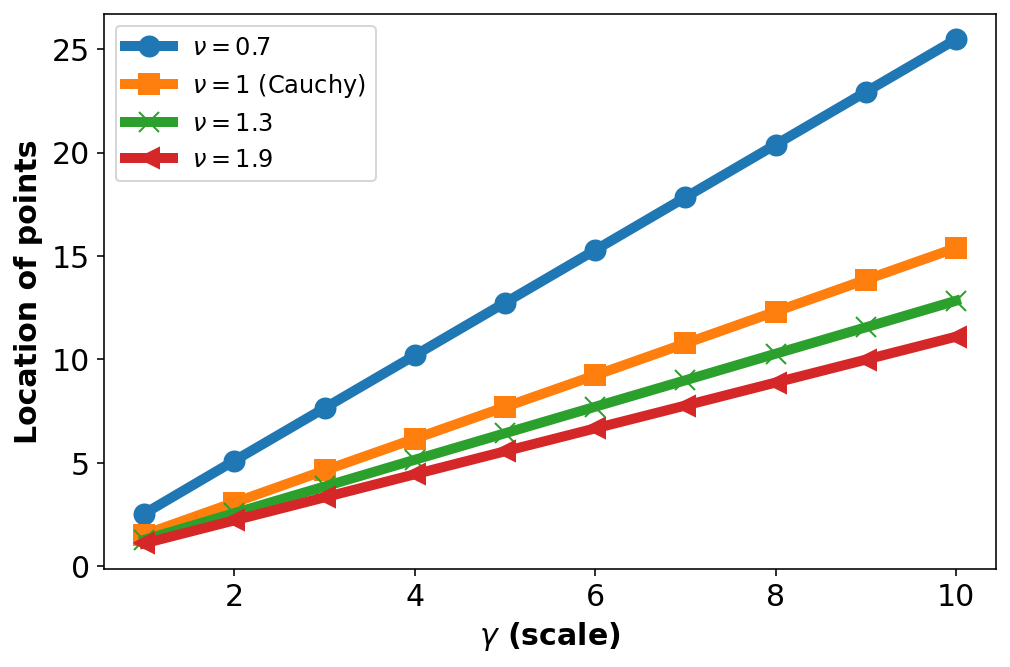}}%
\hfill
\subfloat[]{\includegraphics[width=0.489\textwidth, keepaspectratio]{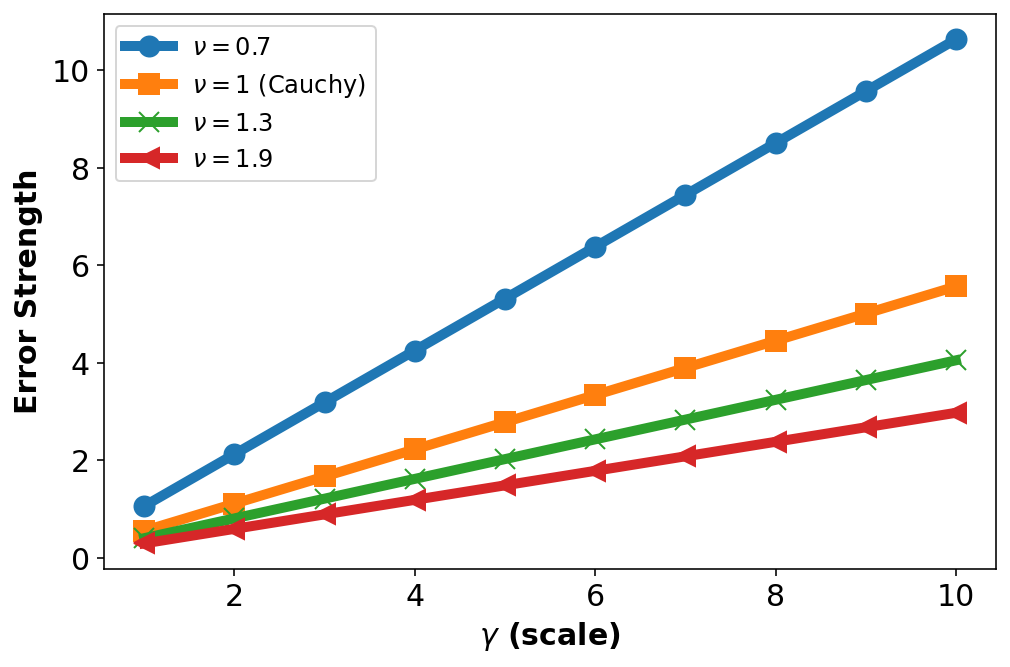}}%
\caption{Optimal $3$-points quantizers whenever $X \sim \mathcal{T}_{\nu}(\gamma)$ for $\nu = 0.7,1,1.3,1.9$. (a) Locations of the positive representation points; negative points are symmetric. (b) Resulting minimum strength.}
\label{fig:RD2}
\end{figure}

We observe that the location of the quantization points and the error strength are both linear with the scale parameter $\gamma_X$. Indeed, we have,
\begin{align}
\underset{\hat{x} }{\arg \min} \, s_\alpha \big( \gamma_X X - \hat{x} \big) & = \underset{ \hat{x}}{\arg \min} \left[ \gamma_X s_\alpha\left(X - \frac{\hat{x}}{\gamma_X}\right) \right] \label{eq:argmin_1} \\
& = \underset{\hat{x} }{\arg \min} \, s_\alpha\left(X - \frac{\hat{x}}{\gamma_X}\right)\label{eq:linear}\\
& = \gamma_X \ \underset{a }{\arg \min} \, s_\alpha(X - a), \label{eq:linear2}
\end{align}
where~\eqref{eq:argmin_1} is due to Property 2-3. Equation~\eqref{eq:linear2} is due to the fact that if $\hat{x}_0$ solves~\eqref{eq:linear} and $a_0$ solves~\eqref{eq:linear2} then $\hat{x}_0 = \gamma_X\,a_0$. Note that equation~\eqref{eq:linear} implies that if $X$ is scaled by $\gamma_X$, then the location of quantization points is also scaled by $\gamma_X$. The resulting minimum error strength scales with $\gamma_X$ as well:
\begin{equation*}
s_\alpha \big( \gamma_X X - \gamma_X\hat{x} \big) = \gamma_X s_\alpha(X- \hat{x}).
\end{equation*}


\subsection{Application}

Consider a communication channel where the input is subject to additive independent Cauchy noise. Under appropriate input constraints, the capacity of this channel was studied and derived in~\cite{FAF2014} where a Cauchy source was found to be optimal. Motivated by real-world applications where only finite-precision devices are available, we consider the {\em quantized\/} input and output of this Cauchy channel. More precisely, a Cauchy source is quantized to 10 points; transmitted over the additive Cauchy-noise channel; and the received output is quantized to 10 points.

Since the source is Cauchy and the output is heavy-tailed (due to the additive Cauchy noise component), we use the proposed strength-based quantizer presented above. Under quantization, the channel becomes a discrete memoryless channel and we quantify the quality of the quantization scheme through the value of mutual information achieved when an input power constraint of 12.01 is imposed,
the power being computed as  $\sum_{\set{R}_i} \Pr(\set{R}_i) \, \hat{x}_i^2$, where the $\{\set{R}_i\}$'s and $\{\hat{x}_i\}$'s denote the quantization regions and representation points respectively. The probabilities of each region can be readily computed as the variable is Cauchy distributed. For comparison, we also consider quantizers optimized 
for $p^{\text{th}}$ absolute moments $\E{|X-\hat{X}|^p}$ (for several values of $p < 1$) and the Maximum Output Entropy quantizer (MOE). Note that the points obtained in all cases are rescaled such that their power matches 12.01. We show in Figure~\ref{fig:mutualinf} the mutual information of the quantized channel function of the scale of the additive Cauchy noise of the channel.
\begin{remark}
As its name indicates, the MOE quantizer is designed so that the quantization regions are equiprobable yielding a maximum entropy at the output of the quantizer. As for the representation points, they are taken midpoints of the regions, except for the boundary points, which are chosen to match the power constraint.
\end{remark}
As seen from the data, the strength-based quantizer outperforms the $p^{\text{th}}$ absolute moment quantizers and closely matches the MOE one, for all noise levels; we conclude that the proposed strength-based quantizer is hence more suited for such heavy-tailed data streams that alternative robust ones. We note that while the proposed quantizer follows from the general framework of lossy data compression, the MOE quantizer does not return specific representation points and is not tied to a specific ``fidelity" measure. We also point out that with a quantized input, the channel output is not Cauchy distributed and our quantizer at the output is not tuned to its statistics. One can consider quantizing only the output of the channel with a continuous Cauchy-distributed input. In that scenario the strength-based quantizer becomes more favorable than the MOE as shown in Figure~\ref{fig:relMI} with a modest relative improvement increasing with the noise scale.

\begin{figure}[!t]
\centering
\includegraphics[width=0.489\textwidth,height=0.315\textwidth]{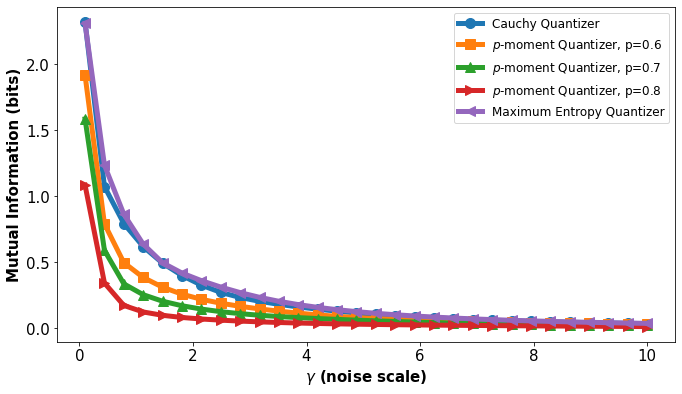}%
\caption{The mutual information $I(X,X+N)$ whenever $N \sim S(1,\gamma_N)$ for different values of $\gamma_N$.}
\label{fig:mutualinf}
\end{figure}

\begin{figure}[!t]
\centering
\includegraphics[width=0.489\textwidth,height=0.315\textwidth]{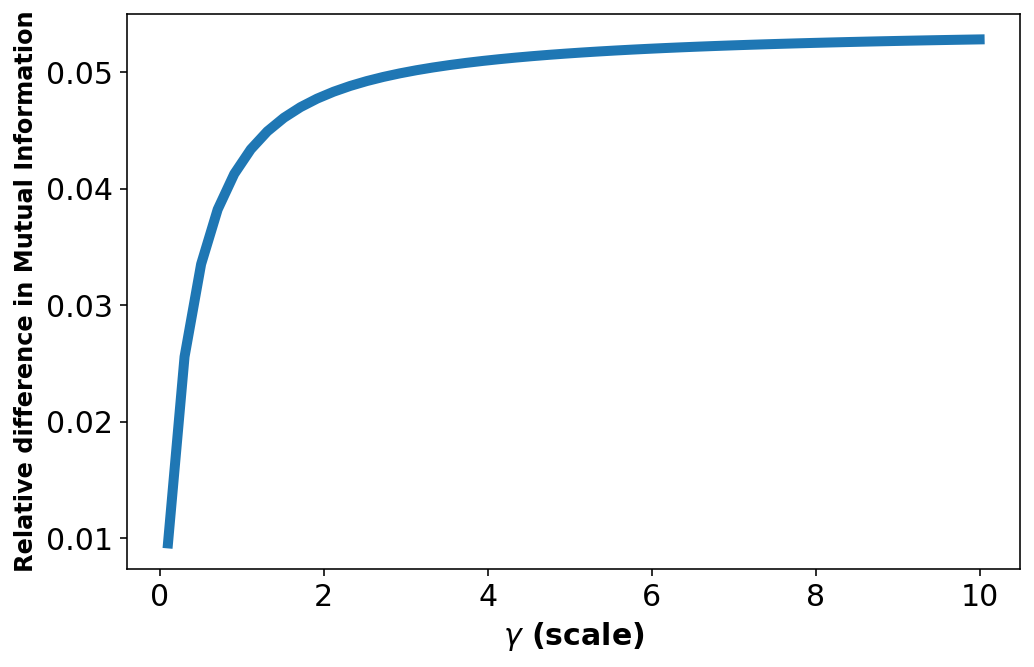}%
\caption{The relative improvement of mutual information of the strength-based quantizer over the MOE quantizer.}
\label{fig:relMI}
\end{figure}

\section{High-Rate Scalar Quantizers}
\label{sc:QuantHR}

For an "entropy-coded quantizer", one imposes a constraint on the entropy of the output of the quantizer instead of choosing a specific $M$. It is widely known that a uniform quantizer is asymptotically optimal~\cite{gallager2008principlesdig} for sources with a finite second moment when using the MSE distortion measure. Naturally, one is interested in whether uniform quantizers are asymptotically optimal in our framework as well. We answer this question first by analyzing the performance of uniform quantizers of a source $X$ when the distortion is measured by $s_\alpha(X - \hat{X})$ at a high rate, i.e. when the length of the quantization regions $|\set{R}_j| = \Delta_j = \Delta \to 0$, for all $j \in \Integers$. Then we prove that uniform quantizers are asymptotically optimal when adopting the strength measure.
\subsection{High Rate Uniform Scalar Quantizers}

We start by formally specifying the uniform quantizer; The representation points are uniformly placed at
\begin{equation*}
    \hat{x}_k = k \Delta, \qquad k \in \Integers,
\end{equation*}
and the quantization regions are defined by mid-points, all of equal width $\Delta$:
\begin{equation*}
    \set{R}_k = \left( k \Delta - \frac{\Delta}{2}, k \Delta +\frac{\Delta}{2} \right], \qquad k \in \Integers.
\end{equation*}

With a view towards the scenario of a ``high-rate regime", $\Delta$ is thought of to be small and we take 

Note that the points are asymptotically optimally located in terms of getting the minimum quantization error strength. Indeed, as $\Delta$ becomes very small, one can approximate the PDF over the $\{\set{R}_j\}$'s by its mean-values
\begin{equation}
\label{eq:probapprox}
f_X(x) \approx  \Bar{f}_j \eqdef \frac{\displaystyle \int_{\set{R}_j} f_X(x) \, dx}{\Delta} = \frac{p_j}{\Delta} \qquad x \in \set{R}_j,
\end{equation}
and
\begin{align}
    h(\tilde{Z}_{\alpha}) = \, & -\E{\log \pz \left(\frac{X-\Xhd}{s_\alpha(X - \Xhd)}\right)}  \\
    = \, & - \sum_{j} \int_{\set{R}_j} f_X(x)\log \pz \left(\frac{x-\hat{x}_j}{s_\alpha(X - \Xhd)}\right) dx \notag \\
    \approx \, & - \sum_{j} \int_{\set{R}_j} \frac{p_j}{\Delta}\log \pz \left(\frac{x-\hat{x}_j}{s_\alpha(X - \Xhd)}\right) dx. \label{eq:UQ_1}
\end{align}
Examining any of the individual terms inside the sum:
\begin{equation}
     - \frac{1}{\Delta} \int_{\set{R}_j}^{} \log \pz \left(\frac{x-\hat{x}_j}{s_\alpha(X - \Xhd)}\right) dx 
     = - \frac{1}{\Delta} \int_{-\frac{\Delta}{2} +\delta}^{\frac{\Delta}{2} + \delta} \log \pz \left(\frac{u}{s_\alpha(X - \Xhd)}\right) du \label{eq:minimsym}
\end{equation}
where $-\frac{\Delta}{2} \leq \delta \leq \frac{\Delta}{2}$. Due to the symmetry of $- \log \pz (x)$ and due to the fact that it is increasing in $|x|$ (Property~\ref{prop}-3), Equation~\eqref{eq:minimsym} is minimized when $\delta = 0$. In summary, for a uniform quantizer in the high-rate regime, the RHS of Equation~\eqref{eq:UQ_1} is minimal whenever the quantization points are middle points of the regions, which in turn leads to minimizing $s_\alpha(X - \Xhd)$ by virtue of Property~\ref{prop:stableprop}-5.

Next, we evaluate the performance of this quantizer. More precisely, we find the error-strength $s_\alpha(X-\hat{X})$ as a function of the uniform length of every quantization interval $|\set{R}_j| = \Delta$, for all $j \in \Integers$. 

We present the main result in the form of Theorem~\ref{th:unif}.
\begin{theorem}
\label{th:unif}
    Let $X$ be a random variable such that $s_{\alpha}(X)$ exists.
Let $\Xhd$ be the quantized version of $X$ with uniform quantization regions of length $\Delta$:
\begin{align*}
    x \in \left( k \Delta - \frac{\Delta}{2}, k \Delta +\frac{\Delta}{2} \right] \longrightarrow \hat{x}_{k} = k \Delta, \quad k \in \Integers.
\end{align*}

Then, as $\Delta \rightarrow 0$,
\begin{equation}
\label{eq:solva}
\frac{s_\alpha(X-\Xhd)}{\Delta} 
\xrightarrow[\Delta \rightarrow 0]{}
s_{\alpha}(U),
\end{equation}
where $U$ is the uniform random variable $\set{U}\left( -\frac{1}{2},\frac{1}{2} \right)$.
\end{theorem}

 Before we proceed with the proof, we note that the existence of $s_{\alpha}(X)$ is guaranteed whenever the conditions of Property~\ref{prop:stableprop} in Appendix~\ref{app:prop} are satisfied.

\begin{proof}
    The performance of this uniform quantizer is measured by $s_\alpha(X - \Xhd)$ and we recall that:
\begin{align}
    -\E{\log \pz \left(\frac{X-\Xhd}{s_\alpha(X - \Xhd)}\right)} = h(\tilde{Z}_{\alpha}). \label{eq:powconst}
\end{align}
As ``the rate" of the quantizer increases, $\Delta$ decreases, and the approximation~\eqref{eq:probapprox} holds. The LHS of Equation~\eqref{eq:powconst} can be approximated as in~\eqref{eq:UQ_1} by: 
\begin{align}
     & - \sum_{j} \int_{\set{R}_j} \frac{p_j}{\Delta}\log \pz \left(\frac{x-\hat{x}_j}{s_\alpha(X - \Xhd)}\right) dx \nonumber \\
    = \, & - \frac{1}{\Delta} \sum_{j} p_j \int_{\hat{x}_j - \frac{\Delta}{2}}^{\hat{x}_j + \frac{\Delta}{2}} \log \pz \left(\frac{x-\hat{x}_j}{s_\alpha(X - \Xhd)}\right) dx  \nonumber\\
    = \, & - \frac{s}{\Delta} \int_{-\frac{\Delta}{2s}}^{\frac{\Delta}{2s}}  \log \pz \left(t\right) dt\label{eq:uniformquant2}\\
    = \, & -\int_{-\frac{1}{2}}^{\frac{1}{2}} \log f_{\tilde{Z}_{\alpha}}\left(\frac{u}{a}\right)\,du \label{eq:finalappox},
\end{align}
where $s \eqdef s_\alpha(X - \Xhd)$, $a \eqdef \frac{s}{\Delta}$ and where we used the change of variable $u = a t$. To write Equation~\eqref{eq:uniformquant2}, we used the change of variable $t = \frac{x - \hat{x}_j}{s}$.  

The approximation in~(\ref{eq:finalappox}) implies that the performance of the uniform quantizer in the high-rate regime is governed by solving for $a = \frac{s}{\Delta} = \frac{s_\alpha(X - \Xhd)}{\Delta}$ the equation 
\begin{equation*}
-\int_{-\frac{1}{2}}^{\frac{1}{2}} \log f_{\tilde{Z}_{\alpha}}\left(\frac{u}{a}\right)\,du = h(\tilde{Z}_{\alpha}),
\end{equation*}
which in turn implies by the definition of $s_{\alpha}(\cdot)$ that 
\begin{equation*}
a = \frac{s_\alpha(X - \Xhd)}{\Delta} = s_\alpha(U),
\end{equation*}
where $U \sim \set{U} \left( -\frac{1}{2},\frac{1}{2} \right)$. 
\end{proof}

Hence the quantization error strength for the optimal quantizer is asymptotically linear in $\Delta$ with slope being equal to $s_\alpha(U)$. Equation~(\ref{eq:solva}) gives a rather ``nice" generalization of the well-known result regarding the performance of uniform quantizers in the high-rate regime under the MSE distortion measure. In fact, considering the case $\alpha = 2$, $s_{\alpha}(\cdot)$ boils down to the standard deviation as seen in~\eqref{eq:stren2} and $s_{2}(U) = \sqrt{ \text{E}[U^2] } = \sqrt{\frac{1}{12}}$, 
which recovers the well-known result for the performance of the uniform quantizer under an MSE distortion~\cite{gallager2008principlesdig}. 

\subsection{Uniform Scalar Quantizers are Optimal in the High Rate Regime}
In this section , we study the performance of non-uniform quantizers under the strength distortion in the high-rate regime. Following a similar sequence of arguments to those made in~\cite{gallager2008principlesdig}, we show that the result obtained in Theorem~\ref{th:unif} for uniform high-rate quantizers provides an approximate lower-bound on the strength of the quantization error for any non-uniform scalar quantizer. 
A similar analysis was carried out in~\cite{GishAndPierce} for difference-based distortion measures, where it was shown that uniform quantizers are optimal in the high rate regime. A key difference, in our setup, is that the strength, which measures the distortion, is defined as an implicit solution to equation~\eqref{eq:powdef} and thus does not fall under the category of difference-based measures considered in~\cite{GishAndPierce}. This new formulation of the distortion measure warrants a separate analysis.

Within the framework of designing entropy-coded quantizers, we impose a constraint on the average rate of the quantizer's output $V$ which is captured by its entropy $H(V)$. We consider the quantization regions $\{ \set{R}_j \}_j$ to be of variable widths $\{ \Delta_j \}_j$ which are not restricted, and we denote by $\{p_j\}_{j \in \Integers}$, the probability mass function of $V$. 

At high-rates, we approximate for $x \in \set{R}_j$
\begin{equation}
\label{eq:apprononuni}
f_X(x) \approx \Bar{f}_j \eqdef \frac{\int_{\set{R}_j} f_X(x)}{\Delta_j} = \frac{p_j}{\Delta_j}, \qquad x \in \set{R}_j.
\end{equation}
We have
\begin{align}
    H(V) & = -\sum_{j}^{} p_j\log p_j
         = -\sum_{j}^{} p_j\log(\Bar{f}_j\Delta_j) \notag \\
          &= -\sum_{j}^{} \left(\int_{\set{R}_j} f_X(x) \, dx\right)\log(\Bar{f}_j\Delta_j) \notag \\
          &= -\sum_{j}^{} \int_{\set{R}_j} f_X(x) \log(\Bar{f}_j\Delta_j)  \, dx \notag \\
          &\approx -\sum_{j}^{} \int_{\set{R}_j} f_X(x) \log(f_X(x) \Delta_j)  \, dx \label{eq:approxdelta1}\\
         &= - \! \int_{-\infty}^{\infty} \!\! f_X(x) \log f_X(x)\,dx  - \! \int_{-\infty}^{\infty} \!\! f_X(x) \log \Delta(x)\,dx \notag\\
         &= h(X) - \int_{-\infty}^{\infty} f_X(x) \log \Delta(x)\,dx \label{eq:entropyConstraint},
\end{align}
where for any $x \in \Reals$, $\Delta(x) = \Delta_j$ if $x \in \set{R}_j$. The approximation in Equation~\eqref{eq:approxdelta1} holds true since $\Bar{f}_j \approx f_X(x)$ for $x \in \set{R}_j$. Since over every region, the PDF is approximately flat, the representation points are optimally placed in the middle of the regions as argued in Section~\ref{sc:QuantHR}.

Using the approximation~\eqref{eq:apprononuni} for the LHS of Equation~\eqref{eq:powconst} as in the uniform case in Section~\ref{sc:QuantHR}, we get
\begin{align}
    h(\tilde{Z}_\alpha) \, & \approx - \sum_{j}^{} \frac{p_j s}{\Delta_j} \int_{-\frac{\Delta_j}{2s}}^{\frac{\Delta_j}{2s}}  \log f_{\tilde{Z}_{\alpha}} \left(x\right) dx  \notag \\
    & = - \sum_{j}^{} p_j  \int_{-\frac{1}{2}}^{\frac{1}{2}}  \log f_{\tilde{Z}_{\alpha}} \! \left(\frac{u \Delta_j}{s}\right) \, du \\
    & = - \sum_{j} \int_{\set{R}_j} f_X(x) \left[ \int_{-\frac{1}{2}}^{\frac{1}{2}}  \log f_{\tilde{Z}_{\alpha}} \! \left(\frac{u \Delta(x)}{s}\right) du\right] dx \notag\\
    & =  - \int_{-\infty}^{\infty} \! f_X(x) \left[ \int_{-\frac{1}{2}}^{\frac{1}{2}}  \log f_{\tilde{Z}_{\alpha}} \!\! \left(\frac{u \Delta(x)}{s}\right) du\right] dx. \label{eq:consnonunif}
\end{align}
Our objective is to find the $\Delta_j$'s that minimize $s$ where $s= s_\alpha(X - \hat{X}_{\Delta(x)})$ is the strength of the quantization error. To this end, we formulate the problem as a constrained optimization problem:
\begin{equation}
    \underset{
    \begin{aligned}
        \scriptstyle & \qquad \scriptstyle \text{Equation}~\eqref{eq:consnonunif}\\
        \scriptstyle h(X) - & \scriptstyle \int_{-\infty}^{\infty} f_X(x) \log \Delta(x)\,dx  & \hspace{-6pt} \leq  \scriptstyle \,\, c
    \end{aligned}
}{\min}s_\alpha(X - \hat{X}_{\Delta(x)}),
\label{eq:Problem}
\end{equation}
for some $c >0$. We construct the Lagrangian:
\begin{align}
&L(s,\Delta(x),\lambda_1,\lambda_2) \notag\\
   &= 
    s + \lambda_1\left[-\int_{-\infty}^{\infty} f_X(x) \left[ \int_{-\frac{1}{2}}^{\frac{1}{2}}  \log f_{\tilde{Z}_{\alpha}} \! \left(\frac{u \Delta(x)}{s}\right) du\right] dx  - h(\tilde{Z}_\alpha)\right] + \lambda_2\left[h(X) - \int_{-\infty}^{\infty} f_X(x) \log \Delta(x)\,dx - c \right] \notag \\ 
    &= s - \lambda_1  h(\tilde{Z}_\alpha) + \lambda_2 (h(X) - c)  
     - \int_{-\infty}^{\infty} f_X(x) \bigg[ \lambda_1\int_{-\frac{1}{2}}^{\frac{1}{2}}  \log f_{\tilde{Z}_{\alpha}}\left(\frac{u \Delta(x)}{s}\right) \, du 
     + \lambda_2 \log \Delta(x)\bigg] \, dx, 
    \label{eq:kkt}
\end{align}
for some $\lambda_1 \in \mathbb{R}$, $\lambda_2 \geq 0$. 
We observe from the inner integral in equation~\eqref{eq:kkt} that if a minimizer exists, then it can be selected in such a way \( \Delta(x) \) is constant with respect to \( x \): \( \Delta(x) = \Delta \) for all \( x \). As a result, equation~\eqref{eq:consnonunif} reduces to
\[
h(\tilde{Z}_\alpha) = - \int_{-\frac{1}{2}}^{\frac{1}{2}} \log f_{\tilde{Z}_\alpha} \left( \frac{u \Delta}{s} \right) \, du,
\]
which is an increasing function of \( \Delta \). Since our goal is to minimize \( s_\alpha(\cdot) \), and since $h(X) - \log(\Delta)$ is decreasing with $\Delta$, we choose
\[
\Delta^* = e^{c - h(X)}
\]
as prescribed by equation~\eqref{eq:entropyConstraint} . A solution to the optimization problem is then given by
\[
h(\tilde{Z}_\alpha) = - \int_{-\frac{1}{2}}^{\frac{1}{2}} \log f_{\tilde{Z}_\alpha} \left( \frac{u}{\frac{s}{\Delta^*}} \right) \, du.
\]

We conclude that
\[
s_\alpha^*(X - \hat{X}_{\Delta(x)}) = s_\alpha(X - \hat{X}_{\Delta^*}) = s_\alpha(U) \cdot \Delta^*,
\]
where \( U \sim \mathcal{U}(-\tfrac{1}{2}, \tfrac{1}{2}) \).

\subsection{Numerical Analysis}
\label{sec:CB}

Strength-based quantizers provide a generalization of the MSE quantizer for heavy-tailed sources. Moreover, the analysis conducted in the previous section shows that the performance of uniform quantizers converges asymptotically to the optimal performance at high rate. This family of quantizers is defined using the PDF of a standard $\alpha$-stable variable, namely $\tilde{Z}_{\alpha} \sim \mathcal{S}\left(\alpha,\left(\frac{1}{\alpha}\right)^{\frac{1}{\alpha}}\right)$. One possible inconvenience is that, except in two special cases, the PDFs $f_{\tilde{Z}_{\alpha}}(\cdot)$ do not have closed-form expressions which might lead to them being less-attractive to use. In this section, we focus on the Cauchy-Strength (CS) quantizer of a source $X$, a strength-based quantizer that uses the PDF of a standard Cauchy variable $\tilde{Z}_{1} \sim \mathcal{S}(1,1)$,  
\begin{equation}
f_{\tilde{Z}_{1}}(x) = \frac{1}{\pi} \frac{1}{1 + x^2},
\label{eq:PDFCauchy}
\end{equation}
for which the strength $s_1(X)$ is given as the unique constant such that (see Equation~\eqref{eq:stren1} for $d = 1$) 
\begin{equation}
\label{eq:s1}
\text{E}\left[\log\left(1 + \frac{X^2}{s_1(X)^2}\right)\right] = \log 4.
\end{equation}

The CS quantizer provides an explicit extension to the MSE quantizer for heavy-tailed infinite second-moment sources $X$. Being a member of the family of strength-based quantizers, all previously presented results are applicable to this specific case. Namely, one can apply the results of Theorem~\ref{th:unif} for the high-rate uniform quantizer. For the CS quantizer, 
$s_{1}(U)$ for $U \sim \mathcal{U}(-\frac{1}{2},\frac{1}{2})$ can be found using~\eqref{eq:s1}
\begin{equation*}
    \int_{-\frac{1}{2}}^{\frac{1}{2}} \log\left(1 + \frac{u^2}{s_1(U)^2}\right) \,du = \log 4,
\end{equation*}
which through integration by parts simplifies to
\begin{equation*}
     \ln\left(\frac{1}{4s_1(U)^2} + 1\right) - 2 
    + 4s_1(U)\arctan\left(\frac{1}{2s_1(U)}\right) = \ln(4),
\end{equation*}
with a solution $s_1(U) = 0.1359$. This means that the rate of convergence of the quantization error strength to $0$ --with $\Delta$,  is approximately twice as fast for a Gaussian-like (thin-tailed) source compared to a heavy-tailed one $\left( \frac{s_{2}(U)}{s_{1}(U)} \approx 2.124 \right)$. Finally, we note that, in the general case, $s_{\alpha}(U)$ can be numerically evaluated for all values of $0 < \alpha \leq 2$ as depicted in Figure~\ref{fig:Alpha-Power}. 

From a practical point of view, we compare the performance of the uniform quantizer in the finite regime with respect to the ultimate lower bound provided by $R(D)$ or equivalently $D(R)$ in Theorem~\ref{th:extstable} for two specific sources; the heavy-tailed Cauchy and the thin-tailed Gaussian. The analysis is conducted as a function of the number of quantization levels $M$ where we plot in Figures~\ref{fig:UniformCauchy} and~\ref{fig:UniformGaussian} the difference between $[s_\alpha(X - \hat{X})- D(R(M))]$ where $s_\alpha(X - \hat{X})$ is numerically evaluated for the uniform quantizer. Furthermore, we plot on the same figures the same difference for $s_\alpha(X - \hat{X})$ as given by Equation~\eqref{eq:solva}. Three important facts are observed in Figures~\ref{fig:UniformCauchy} and~\ref{fig:UniformGaussian}: 
\begin{itemize}
\item First, the plots are decreasing towards zero with the number of points, a fact which indicates the asymptotic optimality of the uniform quantizer in the high-rate regime for Gaussian and heavy-tailed Cauchy.
\item Second, the speed of convergence is different between the two cases. The number of required representation bits to achieve the same gap from $D(R)$ is one order of magnitude higher for the Cauchy compared to the Gaussian with the same scale.  
\item Finally, the values of $M$ at which approximation~\eqref{eq:solva} starts to become valid are much smaller for a thin-tailed distribution such as the Gaussian compared to heavy tailed ones like the Cauchy.
\end{itemize}
\begin{figure}[!ht]
    \begin{center}
    \includegraphics[width=\linewidth]{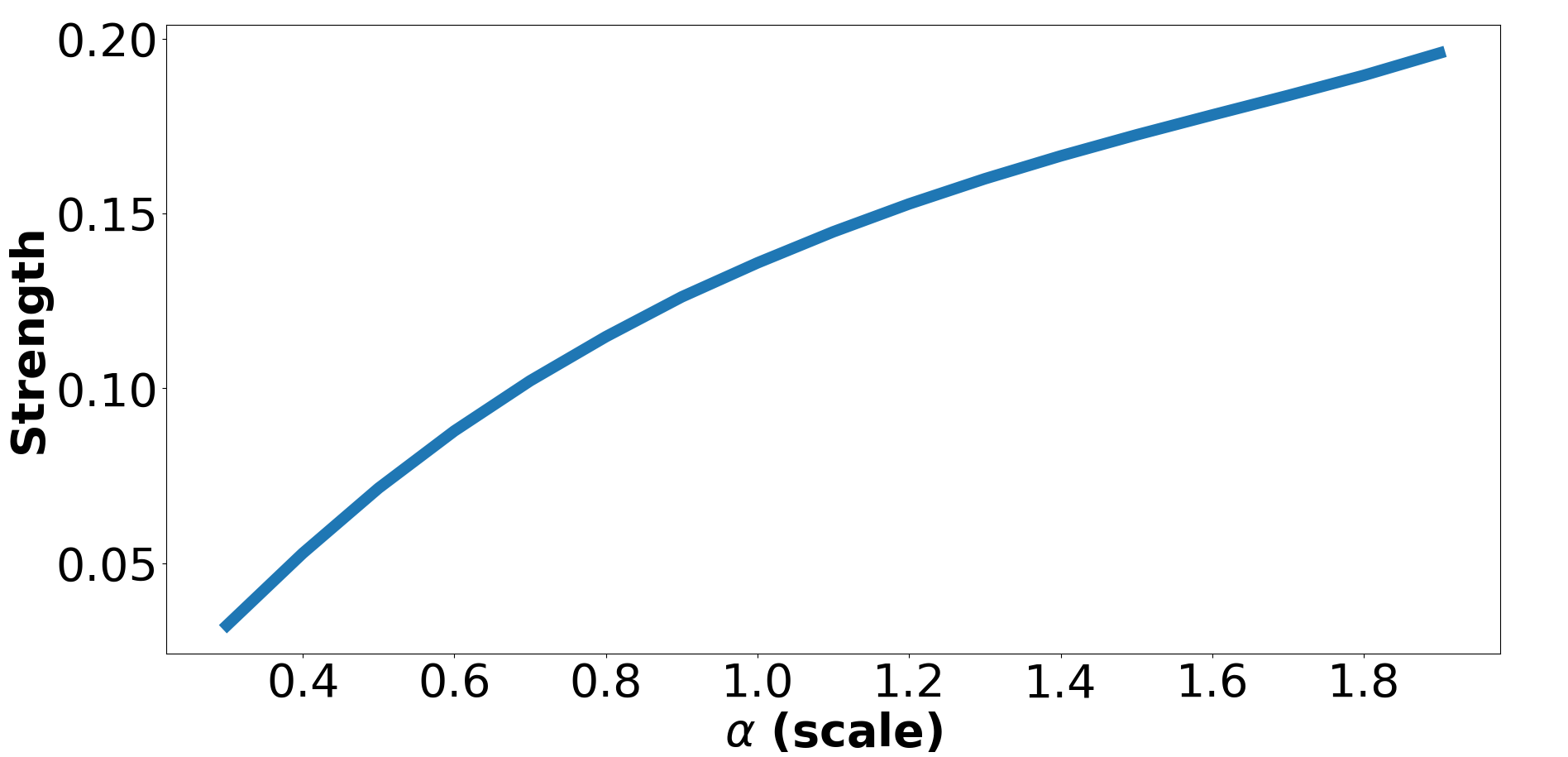}
\caption{$s_\alpha(U)$ for different values of $\alpha$ and $U \sim \set{U} \left( -\frac{1}{2},\frac{1}{2} \right)$.}
    \label{fig:Alpha-Power}
    \end{center}
\end{figure}
\begin{figure}[!ht]
    \begin{center}
    \includegraphics[width=\linewidth]{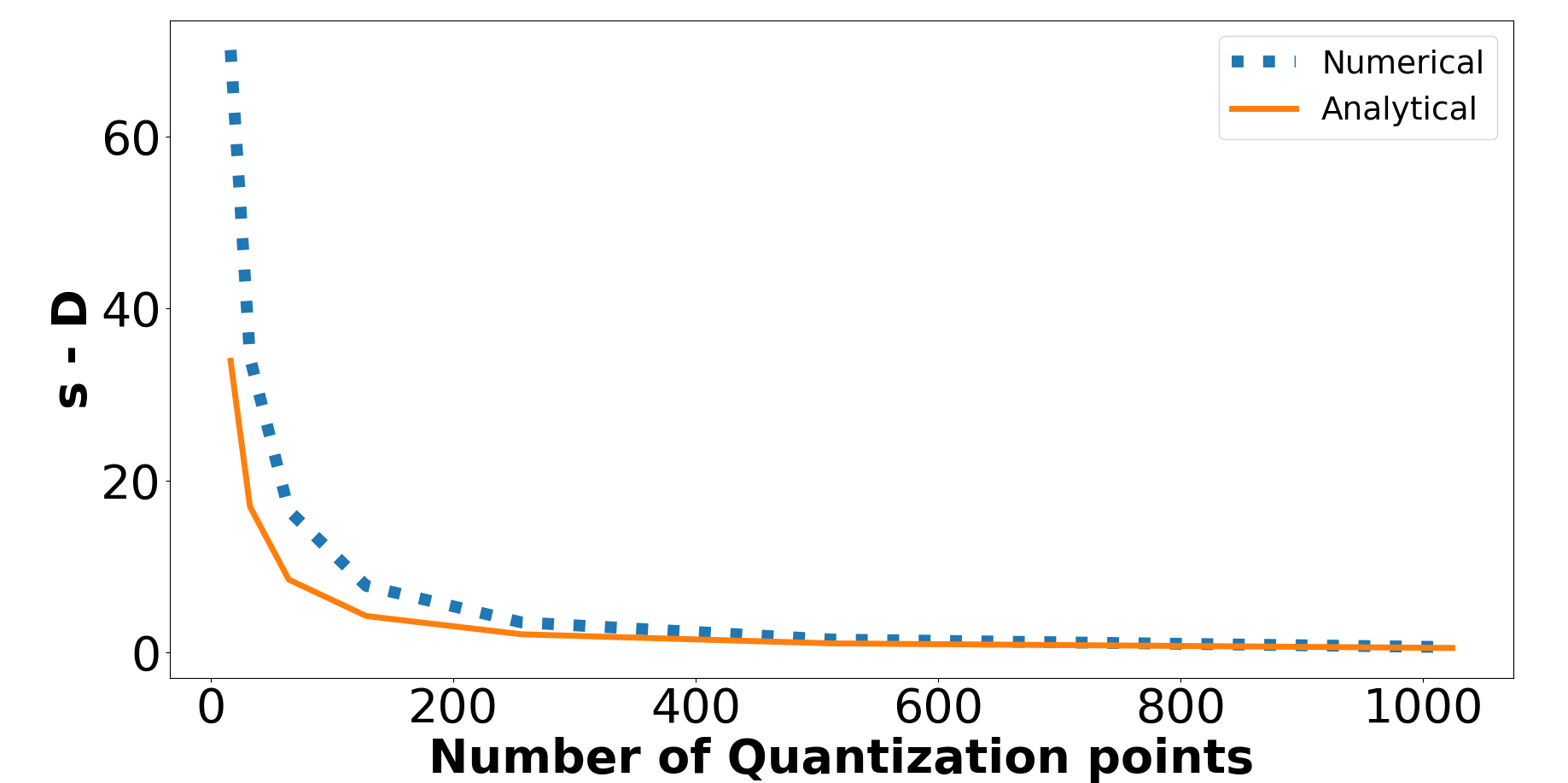}
\caption{$[s_\alpha(X - \hat{X})- D(R)]$ as a function of the uniform quantizer's size $M$ for $X \sim \set{S}(1,1)$. Analytical refers to Equation~\eqref{eq:solva} for $\alpha = 1$.}
    \label{fig:UniformCauchy}
    \end{center}
\end{figure}
\begin{figure}[!ht]
    \begin{center}
    \includegraphics[width=\linewidth]{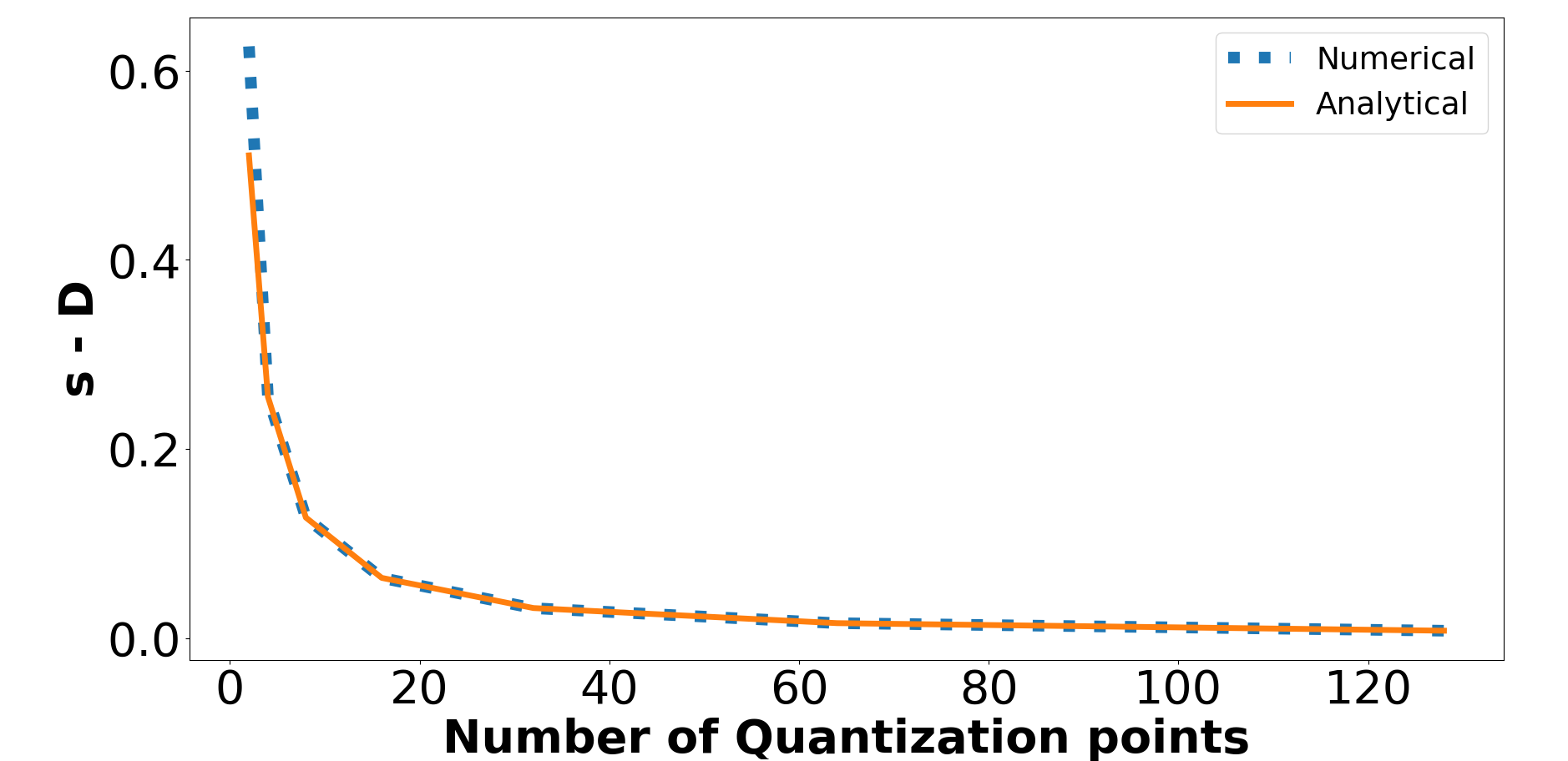}
\caption{$[s_\alpha(X - \hat{X})- D(R)]$ as a function of the uniform quantizer's size $M$ for $X$ a standard Gaussian source. Analytical refers to Equation~\eqref{eq:solva} for $\alpha = 2$.}
    \label{fig:UniformGaussian}
    \end{center}
\end{figure} 


\section{Conclusion}
\label{sc:Concl}

In this paper, we provided a new framework for performing lossy data compression of heavy-tailed sources with substantial implications in the vast areas of digital communications and signal processing. Under this framework, we adopted a {\em strength} measure on the error as a distortion and we derived the rate-distortion function of a stable source for both scalar and vector cases. Our approach finds application, for example, in the quantization problem of heavy-tailed sources, in the sense of specifying the best collection of representation points and quantization regions that minimize the {\em strength} of the quantization error. Moreover, we showed that uniform quantizers remain asymptotically optimal in the high-rate regime and that, in general, there is little advantage in considering non-uniform ones; a fact that is well-known for sources with finite second moment. A key difference is that, in the heavy-tailed case, one has to increase the number of bits of the quantizer to achieve the same distortion levels. This formally proves that standard analog-to-digital converters can still operate well for heavy-tailed sources, but with much higher distortion levels than one would get for Gaussian sources with the same scale. 


\bibliographystyle{IEEEtran}
\bibliography{references,bibliography}
\appendices

\section{Some Properties of Stable Distributions and of $s_{\alpha}(\vX)$}
\label{app:prop}
In this appendix, we list relevant properties of stable distribution in addition to properties of the $\alpha$-power $s_{\alpha}(\cdot)$. 

\begin{properties}
\label{prop}
Let X $\sim \Spdf{\beta}{\gamma}{\delta}$. Then
\begin{enumerate}
    
    \item Let $c \in \Reals$. Then, $X+c \sim \Spdf{\beta}{\gamma}{\delta + c}$. Also,
    \begin{equation*} 
    \begin{array}{lll}
        \displaystyle c X  & \hspace{-7pt}\sim \Spdf{\sgn(c)\beta}{|c|\gamma}{c\delta}, & \text{if } \alpha \neq 1 \\
        \displaystyle c X  & \hspace{-7pt}\sim \Spdf{\sgn(c)\beta}{|c|\gamma}{c\delta - \frac{2}{\pi} c\gamma\beta\ln|c|} & \text{o.w.}
        \end{array}
    \end{equation*}
    
    
    \item When $\alpha \neq 2$, $\E{|X|^p} < \infty$ if and only if $p < \alpha$.

    \item Whenever $X \sim \mathcal{S}(\alpha,\gamma)$ is symmetric S$\alpha$S, then $\pz(\cdot)$ is decreasing with $|x|$. 
    
    \item Let $X_k \sim \Spdf{\beta_k}{\gamma_k}{\delta_k}$, $k = 1, 2$, be two independent stable variables. Then $X_1 + X_2 \sim \Spdf{\beta}{\gamma}{\delta}$, where 
    \[
    \delta = \delta_1 + \delta_2 \quad \beta = \frac{\beta_1\gamma_1^{\alpha}+\beta_2\gamma_2^{\alpha}}{\gamma_1^{\alpha}+{\gamma_2^{\alpha}}} \quad \gamma = (\gamma_1^{\alpha} + \gamma_2^{\alpha})^{\frac{1}{\alpha}}.
    \]
    
    \item Let $\vect{X} \sim \textbf{S}(\alpha,\gamma_{\vect{X}})$ independent of $\vect{Y} \sim \textbf{S}(\alpha,\gamma_{\vect{Y}})$. Then $\vect{X} + \vect{Y} \sim \textbf{S}\left( \alpha,(\gamma_{\vect{X}}^\alpha+\gamma_{\vect{Y}}^\alpha)^{\frac{1}{\alpha}} \right)$.    
\end{enumerate}
\end{properties} 

\begin{proof}
    See~\cite{samorodnitsky1996stable} and~\cite{nolan2020univariate}.
\end{proof} 

\begin{properties}
\label{prop:stableprop}
Let \Xb be a random vector such that:
\[
    \left\{ \begin{array}{ll}
    \E{ \log \left( 1+\left\|\Xb\right\| \right) } < \infty \quad &  \alpha < 2 \vspace{5pt} \\ 
    \E{\|\Xb\|^2} < \infty \quad &  \alpha =2,
\end{array} \right.
\]
and let $\atilde{\vect{Z}} \sim 
\textbf{S} \left(\alpha, \left( \frac{1}{\alpha} \right)^{\frac{1}{\alpha}} \right)$ be {\em a reference\/} $d$-dimensional \SaS vector. Then,
\begin{enumerate}
    \item The strength \( s_\alpha(\Xb) \) exists and except for the deterministically zero vector, it is the unique solution to the equation \( -\E{ \log \, \pzv \left(\frac{\Xb}{s_\alpha(\Xb)}\right) } = h(\atilde{\vect{Z}})\).
    \item \( s_\alpha(\Xb) \geq 0 \) with equality if and only if \( \Xb=0 \).
    \item \( \forall c \in \Reals \), \( s_\alpha(c\Xb)= \left|c\right|s_\alpha(\Xb) \).
    \item Whenever $\Xb \sim \textbf{S}(\alpha,\gamma)$,
        $s_{\alpha}(\Xb) = (\alpha)^{\frac{1}{\alpha}} \gamma$.
    \item $g(P) \eqdef - \E{ \log  \pzv \left(\frac{\Xb}{P} \right) }$ is non-increasing and continuous in \( P \) for $ P > 0 $ .
\end{enumerate}
\end{properties}
\begin{proof}
We refer the reader to \cite{fahs2017information}
\end{proof}



\section{Proofs of Theorems 2,3 and 4}
\label{app:proofs}
In this appendix, we provide the proofs of Theorems \ref{th:extstable}, \ref{th:extToVectors1} and \ref{th:extToVectors2} stated in Section \ref{sc:RD}
\subsection{Proof of Theorem \ref{th:extstable}}
\noindent
\begin{proof}
We first show that 
\begin{equation*}
\phi_{X}(D) = \underset{P_{\hat{X}|X}: s_\alpha(X-\hat{X}) \leq D}{\inf} I(X;\hat{X}),
\end{equation*}
evaluates to \eqref{eq:star}.

When $D \geq s_\alpha(X)$, one can pick $\hat{X}=0$ deterministically as a reproduced character, and we get $I(X;\hat{X})= h(X) - h(X) = 0$ with $s_\alpha(X-\hat{X}) = s_\alpha(X) \leq D$, therefore achieving $\phi_{X}(D) = 0 $.

On the other hand, when $0< D < s_\alpha(X)$, we find a lower bound on $I(X;\hat{X})$ among all $P_{\hat{X}|X}$ such that $s_\alpha(X-\hat{X}) \leq D$ by noting the following:
\begin{align}
    I(X;\hat{X}) &= h(X)-h(X|\hat{X})\notag \\ 
    &= h(X)-h(X-\hat{X}|\hat{X})\notag \\
    &\geq h(X)-h(X-\hat{X})  \label{eq:condent}\\
    &= h(s_\alpha(X)\tilde{Z}_\alpha)-h(X-\hat{X}) \label{eq:hEqual} \\
    &\geq  h(s_\alpha(X)\tilde{Z}_\alpha)-h(s_\alpha(X-\hat{X})\tilde{Z}_\alpha) \label{eq:entmax}\\
    &=\log \left[ \frac{s_\alpha(X)}{s_\alpha(X-\hat{X})} \right] \notag \\
    &\geq \log \left[ \frac{s_\alpha(X)}{D} \right] = \frac{1}{\alpha}\log(\alpha) + \log\left(\frac{\gamma_X}{D}\right), \label{eq:endres}
\end{align}
where equation~(\ref{eq:condent}) is due to the fact that conditioning cannot increase entropy. Equation~\eqref{eq:hEqual} is substantiated by
\begin{align*}
    s_{\alpha} (\atilde{Z}) & = 1 \\
    \implies s_{\alpha} ( s_\alpha(X) \, \atilde{Z} ) & = s_\alpha(X) \\
    \implies \hspace{5pt} h ( s_\alpha(X) \atilde{Z} ) & = h(X),
    \quad \text{when } X \sim \text{\SaS},
\end{align*}
through Properties~\ref{prop:stableprop} of the strength measure (See Appendix~\ref{app:prop}).
Equation~(\ref{eq:entmax}) follows from the fact that stable distributions are entropy maximizers under a strength constraint~\cite[Theorem 1]{fahs2017information}. Finally, the inequality in~(\ref{eq:endres}) follows from $s_\alpha(X-\hat{X}) \leq D$ and the fact that $\log(\cdot)$ is non-decreasing and uses the fact that $s_{\alpha}(X) = (\alpha)^{\frac{1}{\alpha}} \gamma_X$ (see Property~\ref{prop:stableprop}-4 in Appendix~\ref{app:prop}). 

 The lower bound in~(\ref{eq:endres}) is achieved by choosing $\hat{X}$ such that $ X= \hat{X} + Z$ where $Z$ and $\hat{X}$ are independent \SaS with $s_\alpha(Z) = s_\alpha(X-\hat{X}) = D$  and $s_\alpha(\hat{X})= \left(s_\alpha(X)^{\alpha}-D^{\alpha}\right)^{\frac{1}{\alpha}}$ (see Property~\ref{prop}-4). 
Hence, $\phi_{X}(D)$ evaluates to~\eqref{eq:star}.

It remains to show that $R(D) = \RID=\phi_{X}(D)$. The considered setup departs from the conventional one in that the distortion constraint is imposed on $s_\alpha(X-\hat{X})$ and is not of the usual form $\E{d(X, \hat{X})}$. The distortion ``strength" $s_{\alpha}(\cdot)$ is an operator on RVs and not an average. This results in two main challenges:
\begin{itemize}
\item The first one pertaining to the convexity requirement of $\phi_{X}(D)$ which is crucial to the proof of the converse, as highlighted after the statement of Theorem~\ref{th:RDT}. In this work, the distortion constraint is imposed on $s_\alpha(X-\hat{X})$ which is neither linear nor necessarily convex in the distributions, and so the convexity of $\phi_X(D)$ is not a given and must be explicitly verified.

\smallskip
\item Secondly, the proof of the direct part $R(D) \leq \RID$ generally involves a random coding argument that relies on a per-character distortion $d(x,\hat{x})$ that is extended to sequences using separability (see~\cite[p.425-427]{Polyanskiy_Wu_2024} for more details).
However, with the distortion being the ``strength" $s_{\alpha}(\cdot)$ of the error, the standard techniques are no longer suitable. 
\end{itemize}

In order to provide a proof that leverages the results of Theorem~\ref{th:RDT}, we make the following observation: 
\begin{equation}
\label{eq:equiv}
s_\alpha(X - \hat{X}) \leq D \iff -\E{\log \pz \left(\frac{X-\hat{X}}{D}\right)} \leq h(\tilde{Z}_{\alpha}).
\end{equation}
Indeed, the ``$\Longleftarrow$" direction simply follows from the definition of $s_\alpha(\cdot)$ (Definition~\ref{def:str}). On the other hand, the ``$\Longrightarrow$" direction is due to the fact that $g(P) =  -\E{\log \pz \left(\frac{X-\hat{X}}{P}\right)}$ is non-increasing and continuous in $P > 0$ as stated in Property~\ref{prop:stableprop}-5. Hence, an equivalent problem to the one at hand is the following: 

\noindent
Let $0 < D < s_{\alpha}(X)$,
\begin{enumerate}[label = \textbullet]
    \item Define $\tilde{d}(x,\hat{x},D) = -\log \pz \left(\frac{x-\hat{x}}{D}\right)$, $\tilde{d}(x^n,\hat{x}^n,D) = \frac{1}{n} \sum_{i=1}^{n} -\log \pz \left(\frac{x_i-\hat{x}_i}{D}\right)$.
    \item  Pick some $P_{\hat{X}|X}$ such that $-\E{\log \pz \left(\frac{X-\hat{X}}{D}\right)} \leq h(\tilde{Z}_{\alpha})$.
\end{enumerate}
By defining
\[
\Tilde{\phi}_X(D)  \eqdef \underset{P_{\hat{X}|X}: -\E{\log \pz \left(\frac{X-\hat{X}}{D}\right)} \leq h(\tilde{Z}_{\alpha})}{\min} I(X;\hat{X}),
\]
we obtain from equation~\eqref{eq:equiv} that
\[
\phi_X(D) = \Tilde{\phi}_X(D).
\]
Furthermore, we have
\[
\RID = \tRID = \underset{n \to \infty}{\limsup} \ \frac{1}{n} \Tilde{\phi}_{X^{n}}(D).
\]
In the achievability part, we choose \( \hat{X}^n \) such that each \( \hat{X} \) achieves \( \phi_X(D) \). As a result, we choose \( \hat{X}^n \overset{\textup{i.i.d.}}{\sim} P_{\hat{X}} \), where
\[
s_\alpha(\hat{X}) = \left( s_\alpha(X)^{\alpha} - D^{\alpha} \right)^{\frac{1}{\alpha}}.
\]
This IID choice allows us to extend equation~\eqref{eq:equiv} to
\begin{equation}
\label{eq:equivExt}
\frac{1}{n} \sum_{i=1}^{n} s_\alpha(X_i - \hat{X}_i) \leq D
\iff 
- \frac{1}{n} \sum_{i=1}^{n} \E{\log \pz \left(\frac{X_i - \hat{X}_i}{D}\right)} \leq h(\tilde{Z}_{\alpha}).
\end{equation}
As such, any \((n, M, D)\)-code constructed for the distortion measure \( \tilde{d}(\cdot, \cdot, D) \) will also work for \( s_\alpha(\cdot) \), and it follows that
\[
R(D) = \tilde{R}(D).
\]
We make the following observations:
\begin{enumerate} [label = \textbullet]
    \item By definition $\tilde{d}(\cdot,\cdot,D)$ is separable. Moreover, for all $x, y$,
    \begin{equation*}
        - \log  \pz \left(\frac{x-y}{D}\right) \geq - \log \pz\left( 0 \right) \quad (> -\infty)
    \end{equation*} since $\pz(\cdot)$ is maximized at $0$. 
    
    One can simply subtract $- \log \pz\left( 0 \right)$ from $\tilde{d}(\cdot,\cdot)$ to make it non-negative.
    \item Examining \eqref{eq:star},
    \begin{align*}
    D_0 & = \inf \big\{ D : \tilde{\phi}_{X}(D) < \infty \big\} \\ 
    & =  \inf \big\{ D : {\phi}_{X}(D) < \infty \big\} = 0.
    \end{align*}
    \item  $D_{\max} =
    \underset{\hat{x}}{\inf}
    \, \left\{ \E{\tilde{d}(X,\hat{x},D)} \right\} \leq \E{\tilde{d}(X,0,D)} < \infty$.
    \item $\tilde{\phi}_X(D) = \phi_X(D)$ is convex in $D$ as can be clearly seen in the expression of $\phi_X(D)$ which was shown to be equal to \eqref{eq:star}. 
\end{enumerate}
Hence, we can invoke Theorem~\ref{th:RDT} so that $\tilde{R}(D) = {\tRID} = \Tilde{\phi}_X(D)$ implying $R(D) = \RID = \phi_X(D)$.
\end{proof}

\subsection{Proof of Theorem \ref{th:extToVectors1}}
\begin{proof}
To prove that $R(D) = \phi_{\Xb}(D)$, one can check that the conditions of Theorem 1 hold as done in the scalar case and that $\phi_{\Xb}(D)$ is a convex function of $D$. Showing that the expression of  $\phi_{\Xb}(D) = \underset{P_{\hat{\Xb}|\Xb}: s_\alpha(\Xb-\hat{\Xb}) \leq D}{\inf} I(\Xb;\hat{\Xb})$ evaluates to \eqref{eq:starstar}
and expressing the distortion framework at hand in the form~\eqref{eq:equiv} --compatible with Theorem~\ref{th:RDT}, can be done in an identical fashion to what was done in the proof of Theorem~\ref{th:extstable}.

\end{proof}
\subsection{Proof of Theorem \ref{th:extToVectors2}}
\begin{proof}
As in the previous theorem, one can check that the conditions of Theorem 1 hold by using the equivalency in~\eqref{eq:equiv} in addition to checking  that $\phi_{\Xb}(D)$ is convex in $D$ and we will have $R(D)=\phi_{\Xb}(D)$. 

In what follows, we derive an expression for $\phi_{\Xb}(D)$, by finding an achievable lower bound for $I(\Xb; \hat{\Xb})$ among all $P_{\hat{\Xb}|\Xb}$ such that $ \sum^{d}_{i=1} s_\alpha(X_i-\hat{X}_i) \leq D$:
    \begin{align}
    I(\Xb;\hat{\Xb}) &= h(\Xb)-h(\Xb|\hat{\Xb})\notag \\
                     &= \sum_{i=1}^{d} h(X_i) + \sum_{i=1}^{d} h(X_i|\hat{\Xb},X_1,...,X_{i-1})\notag \\
                     &\geq \sum^{d}_{i=1} h(X_i) -\sum^{d}_{i=1} h(X_i|\hat{X}_i) \label{eq:iid}\\
                     &= \sum^{d}_{i=1} I(X_i;\hat{X}_i) \notag\\
                     &\geq \sum^{d}_{i=1} R(D_i) \label{eq:scalarcase}\\
                     &= \sum^{d}_{i=1} \max \left\{ \log \left[ \frac{s_\alpha(X_i)}{D_i} \right], 0 \right\} \notag
    \end{align}
where $D_i = s_\alpha(X_i-\hat{X}_i)$. We can achieve equality in equation~(\ref{eq:iid}) by choosing $p \left( x^d | \hat{x}^d \right) = \prod^{d}_{i=1} p(x_i|\hat{x}_i)$ and in equation~(\ref{eq:scalarcase}) by choosing the distribution of each $\hat{X}_i \sim S(\alpha,\gamma_{\hat{X}_i})$ as per the results of Theorem~\ref{th:extstable}. The value of each $\gamma_{\hat{X}_i}$, $1 \leq i \leq d$ is chosen such that $D_i = s_\alpha(X_i-\hat{X}_i)$ matches its prescribed value by the following optimization scheme:
\begin{equation}
\phi_\Xb(D) = \underset{\Sigma D_i = D}{\min} \left[ \sum^{d}_{i=1} \max \left\{ \log \left[ \frac{s_\alpha(X_i)}{D_i} \right], 0 \right\} \right].
\label{eq:WFPb}
\end{equation}

In \cite[p.313-315]{Cover2006}, it is shown that the solution to such a problem gives rise to a kind of reverse water-filling solution which is exactly \eqref{eq:starstarstar}, an expression that is convex in $D$. Indeed, let $\{D^*_{0 i} \}_i$ and $\{ D^*_{1i} \}_i$ be the optimal solutions of~\eqref{eq:WFPb} for distortions $D_0$ and $D_1$ respectively. For any $0 \leq \theta \leq 1$ denote by $D_\theta = \theta D_1 + (1-\theta)D_0$. Since $\sum_{i=1}^{d} D^*_{0i} = D_0$ and $\sum_{i=1}^{d} D^*_{1i} = D_1$, then $\sum_{i=1}^{d} \theta D^*_{1i} + (1-\theta) D^*_{0i} = D_\theta$ and $\{\theta D^*_{1i} + (1-\theta) D^*_{0i}\}_i$ are feasible for $D_\theta$. This implies that the set over which we are minimizing is convex (in fact, it is linear). Furthermore,

\begin{align}
&\theta \phi_X(D_1) + (1-\theta)\phi_X(D_0) \notag \\
= \,\, &\sum^{d}_{i=1}  \theta \log \left[ \frac{s_\alpha(X_i)}{D^*_{1i}} \right] +  (1-\theta)\log \left[ \frac{s_\alpha(X_i)}{D^*_{0i}} \right] \notag \\
= \,\, &\sum^{d}_{i=1} \log(s_\alpha(X_i)) - \theta \log \left({D^*_{1i}} \right) -  (1-\theta)\log \left({D^*_{0i}} \right) \notag \\
\geq \,\, & \sum^{d}_{i=1}  \log(s_\alpha(X_i)) - \log \left(\theta D^*_{1i} + (1-\theta)D^*_{0i} \right) \label{eq:log} \\
= \,\, & \sum^{d}_{i=1}  \log \left[ \frac{s_\alpha(X_i)}{\theta D^*_{1i} + (1-\theta)D^*_{0i}} \right]
\geq \,\, \phi_X(D_\theta), \label{eq:WF}
\end{align}
where the inequality in~\eqref{eq:log} is due to the convexity of $-\log(.)$ and in ~\eqref{eq:WF} since $\phi_{\vX}(D)$ is the minimum for all feasible $\{D_i\}_i$'s.

\end{proof}

\end{document}